\documentclass[aps,prb,superscriptaddress,notitlepage,nopacs,amsmath,amstex,amssymb,citeautoscript,longbibliography,floatfix,letter,twocolumn]{revtex4-2}
\usepackage{graphicx}
\usepackage{macros}
\usepackage{subfigure}
\usepackage{adjustbox}
\usepackage{bm}
\usepackage{ulem}
\usepackage{color}
\usepackage{braket}
\usepackage{standalone}
\usepackage{multirow}
\usepackage{tikz}
\usepackage{mathrsfs}
\usepackage{dsfont}
\usepackage{comment}
\usepackage{amsmath, amsthm}
\newtheorem{thm}{Theorem}
\newtheorem{cor}{Corollary}

\usepackage[colorlinks,bookmarks=true,citecolor=blue,linkcolor=red,urlcolor=blue]{hyperref}
\usepackage{cleveref}
\usepackage{orcidlink}

\graphicspath{{./figures/}}
\begin{document}

\title{Exact volume-law entangled zero-energy eigenstates in a large class of spin models}

\author{Sashikanta Mohapatra\orcidlink{0009-0001-0837-9604}}
\email{sashikanta@imsc.res.in}
\affiliation{Institute of Mathematical Sciences, CIT Campus, Chennai 600113, India}
\affiliation{Homi Bhabha National Institute, Training School Complex, Anushaktinagar, Mumbai 400094, India} 

\author{Sanjay Moudgalya}
\email{sanjay.moudgalya@gmail.com}
\affiliation{School of Natural Sciences, Technische Universit\"{a}t M\"{u}nchen (TUM), James-Franck-Str. 1, 85748 Garching, Germany}
\affiliation{Munich Center for Quantum Science and Technology (MCQST), Schellingstr. 4, 80799 M\"{u}nchen, Germany}

\author{Ajit C. Balram\orcidlink{0000-0002-8087-6015}}
\email{cb.ajit@gmail.com}
\affiliation{Institute of Mathematical Sciences, CIT Campus, Chennai 600113, India}
\affiliation{Homi Bhabha National Institute, Training School Complex, Anushaktinagar, Mumbai 400094, India}

\date{\today}

\begin{abstract}
Exact solutions for excited states in non-integrable quantum Hamiltonians have revealed novel dynamical phenomena that can occur in quantum many-body systems. This work proposes a method to analytically construct a specific set of volume-law-entangled zero-energy exact excited eigenstates in a large class of spin Hamiltonians. In particular, we show that all spin chains that satisfy a simple set of conditions host exact volume-law zero-energy eigenstates in the middle of their spectra. Examples of physically relevant spin chains of this type include the transverse-field Ising model, PXP model, spin-$S$ $XY$ model, and spin-$S$ Kitaev chain. Although these eigenstates are highly atypical in their structure, they are thermal with respect to local observables. Our framework also unifies many recent constructions of volume-law entangled eigenstates in the literature. Finally, we show that a similar construction also generalizes to spin models on graphs in arbitrary dimensions.
\end{abstract}

\maketitle

\textbf{\textit{Introduction.}} 
The emergence of thermodynamics from unitary dynamics in isolated quantum many-body systems~\cite{deutsch1991quantum, rigol2008thermalization} provides a fascinating puzzle in contemporary physics. At the heart of this enigma lies the Eigenstate Thermalization Hypothesis (ETH)~\cite{srednicki1994chaos, rigol2012alternatives, deutsch2018eigenstate, kim2014testing, polkovnikov2011colloquium}, which posits that all the eigenstates of a quantum many-body system exhibit thermal properties for macroscopic observables. Thus, a comprehensive understanding of thermalization necessitates venturing beyond the traditional focus on the ground- and low-lying-excited- state properties of quantum many-body systems, and developing a thorough understanding of the structure of states in the bulk of their energy spectrum. Although the special class of integrable systems~\cite{kinoshita2006quantum, calabrese2011quantum, cassidy2011generalized, vidmar2016generalized} offers a pathway to obtain solutions for excited states owing to the existence of extensive conserved quantities in them, they paradoxically serve as prime examples wherein the ETH is violated.

Obtaining analytical expressions for excited eigenstates remains a formidable challenge for generic non-integrable systems. This daunting task has spurred a plethora of research to identify tractable non-integrable models where a description beyond the low-energy states can be given~\cite{Caspers1982some, Arovas1989two, yang1989eta, Vafek2017entanglement, shiraishi2017systematic, Moudgalya2018exact, Iadecola2019exactlocalized, lin2019exact, schecter2019weak, McClarty2020disorder, lin2020quantum, Mark2020exacteigenstates, Yoshida2022exact, Iversen2022Escaping}. This quest was further fueled by the fact that many of these exact eigenstates violate the ETH, leading to anomalous non-thermal dynamics similar to those seen in experiments realizing the so-called PXP model~\cite{bernien2017probing,  su2022observation}. Hence, these eigenstates were dubbed as quantum many-body scars (QMBS)~\cite{Turner2018quantum, Turner2018weak, moudgalya2022quantum, serbyn2021quantum, chandran2022review}. Examples of exact QMBS include towers of quasiparticle states in the Affleck-Kennedy-Lieb-Tasaki spin chain~\cite{AKLT_model_1987, Moudgalya2018AKLT, mark2020unified, moudgalya2020large}, $\pi{-}$bimagnon states in the spin-$1$ $XY$ model~\cite{schecter2019weak, chattopadhyay2020quantum}, $\eta{-}$pairing states in the Hubbard model~\cite{moudgalya2020eta, mark2020eta}, isolated area-law states in a variety of models~\cite{shiraishi2017systematic} including the PXP~\cite{lin2019exact} and others~\cite{lee2020exact, banerjee2021quantum, biswas2022scars,  wildeboer2022quantum, Mizuta2020Exact, Tang2022Multimagnon, Surace2021Exact, Schindler2022Exact, Gotta2022Exact, Langlett2021Hilbert, Iversen2024Tower}. Multiple constructions that produce models hosting exact QMBS eigenstates have also been proposed~\cite{shiraishi2017systematic, moudgalya2020eta, mark2020eta, odea2020from, pakrouski2020many, pakrouski2021group, ren2020quasisymmetry, ren2021deformed, moudgalya2022exhaustive, rozon2023broken}.

However, most of these QMBS states exhibit sub-volume entanglement entropy in contrast to the volume-law predicted by ETH. For some time, the main exception was a class of QMBS called rainbow scars~\cite{langlett2022rainbow}, which do exhibit volume entanglement for some bipartitions, but the expectation value of appropriate local observables in it is still non-thermal. However, recent results have showcased exact volume-law entangled states hosted in various PXP-type models~\cite{ivanov2024volume} and several spin-$1/2$ chains with nearest-neighbor interactions~\cite{udupa2023weak, chiba2024exact}, which are ``thermal" for local observables. These examples are few and far between and have been demonstrated in specific spin models, and the precise conditions for the existence of such solvable states are still unclear.

This letter provides a generic approach and explicit sufficient conditions for identifying a class of exact excited eigenstates in a large class of spin models. We show that any translation-invariant Hamiltonian whose every term formally anti-commutes with an on-site invertible operator times complex conjugation possesses these eigenstates. We start by constructing an exact zero energy eigenstate for spin-$S$ chains with periodic boundary conditions (PBC) and then generalize it to higher dimensions. We further establish that these states show a volume-law scaling of entanglement entropy as predicted by ETH for all contiguous bipartitions. Moreover, these states exhibit thermal properties for any local observable, although they are atypical and athermal for non-local (but few-body) observables. Our construction recovers previously proposed volume-entangled states in various spin-$1/2$ models~\cite{udupa2023weak, ivanov2024volume, chiba2024exact}. Additionally, we construct analytical expressions of many new exact zero-energy states for a diverse set of spin-$S$ Hamiltonians in arbitrary dimensions for both non-integrable and integrable models.

\begin{figure}
    \centering
    \includegraphics[scale=0.15]{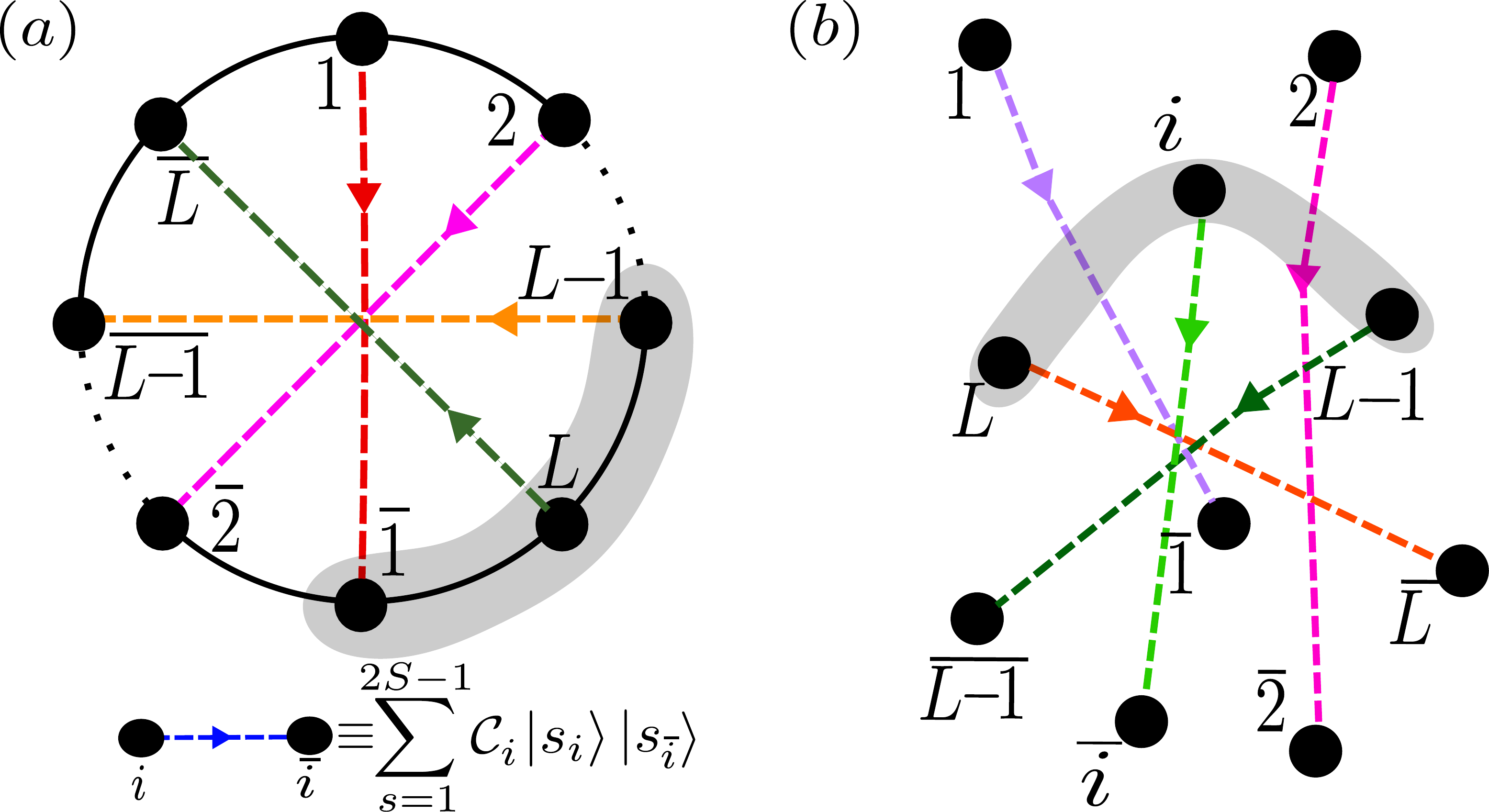}
\caption{Schematic illustration of the exact zero-energy eigenstate $|{\Lambda}{\rangle}_{\mathcal{C}}$ [defined in Eq.~\eqref{eq: zero_energystate_1d}] for a spin-$S$ system of $2L$ sites: $(a)$ on a periodic chain, and $(b)$ in arbitrary dimension or for amorphous systems. The colored dashed lines indicate entangled ``dimers" between spins at sites $i$ and $\overline{i}$. The entanglement in these dimers varies from $0$ to $\text{log}(2S{+}1)$, depending on the on-site invertible operator $\mathcal{C}_{i}$. For example, in the case of spin-$1/2$ and unitary $\mathcal{C}_{i}$, the dashed lines represent pairs of spins in one of the Bell states \{$|\Psi^{\pm}\rangle$,$|\Phi^{\pm}\rangle$\} [defined in Eq.~\eqref{eq: bell states}]. An arbitrary contiguous subsystem in (a) or the collection of sites in (b) (denoted by the shaded region) of size $N{\leq}L$ cuts exactly $N$ correlated bonds leading to an entanglement entropy $\mathcal{S}_{N}{\propto}N$.}
    \label{fig:1}
\end{figure}

\textbf{\textit{General construction in spin chains.}}
We are interested in the class of spin-$S$ chains of size $L$ with a translation invariant interaction (with PBC) of the form
\begin{equation}
    H(L){=}\sum_{\alpha,i{=}1}^{L} \hat{H}^{\alpha}_{i{+}j_{1},\cdots,i{+}j_{n}}{=}\sum_{\alpha,i{=}1}^{L}J_{\alpha}\prod_{k{=}1}^{n} (S_{i{+}j_{k}}^{\mu_k})^{q_{k}},
    \label{eq: generic_Hamiltonian}
\end{equation}
where $S_{i{+}j_{k}}^{\mu_{k}}(\mu_{k}{ \in }\{x,y,z\})$ is the spin-$S$ operator at site $i{+}j_{k}$ and $q_{k}$ is a non-negative integer. The term $\hat{H}^{\alpha}_{i+j_{1},{\cdots}, i{+}j_{n}}$ represents the $\alpha$-type (determined by $\{\mu_{k}\}$ and $\{q_{k}\}$) of interaction between $n{-}$spins at the sites $\{i{+}j_{1},{\cdots}, i{+}j_{n}\}$ ($0{\leq}j_{k}{<}L$), with coupling strength $J_{\alpha}$. The Hamiltonian is defined in the Hilbert space $\mathcal{H}_{L}$ of dimension $\mathcal{D}_{L}{=}(2S{+}1)^{L}$ on the global computational basis $|\vec{S}\rangle{=}{\bigotimes}_{i{=}1}^{L}|s_{i}\rangle$ $(\mathcal{H}_{L}{\equiv}\text{span}\{|\vec{S}\rangle\})$, where $|s_{i}\rangle$ denotes the $(2S{+}1)$ dimensional local Hilbert space on site $i$. In this setting, we establish the following theorem.
\begin{thm}
\label{theorem: 1}
If each $\hat{H}^{\alpha}_{i{+}j_{1},{\cdots}, i{+}j_{n}}$ in the Hamiltonian $H(L)$ anti-commutes with an operator $\mathcal{C}\mathcal{K}$, where $\mathcal{C}{=}\prod_{i{=}1}^{L}\mathcal{C}_{i}$ with $\mathcal{C}_{i}$ on-site invertible operators, and $\mathcal{K}$ is the complex conjugation operator, i.e.,  
\begin{equation}
    \{\hat{H}^{\alpha}_{i+j_{1},\cdots,i+j_{n}},\mathcal{C}\mathcal{K}\}{=}0 \hspace{0.15cm} \text{for}\hspace{0.15cm} i{\in}\{1,{\cdots}, L\},
    \label{eq: anti_commuation}
\end{equation}
then the state [illustrated schematically in Fig.~\ref{fig:1}$(a)$]
\begin{equation}
\label{eq: zero_energystate_1d}
\begin{split}
    |\Lambda\rangle_{\mathcal{C}}{=}{\sum\limits_{|\vec{S}\rangle\in\mathcal{H}_{L}}}\mathcal{C}|\vec{S}\rangle_{1,{\cdots},L}  {\otimes} |\vec{S}\rangle_{\overline{1},{\cdots},\overline{L}}
    {=}{\bigotimes_{i{=}1}^{L}}{\sum_{s{=}1}^{2S{+}1}}\mathcal{C}_{i}|s_{i}{\rangle}|s_{\overline{i}}{\rangle},
\end{split}
\end{equation}
where $\overline{i}{\equiv}i{+}L$, is an exact zero-energy eigenstate of the Hamiltonian $H(2L)$ defined in Eq.~\eqref{eq: generic_Hamiltonian} on a system of size $2L$ with PBC.
\end{thm}
The anti-commutation relation in Eq.~\eqref{eq: anti_commuation} implies
\begin{equation}
    \mathcal{C}(\hat{H}^{\alpha}_{i{+}j_{1}, \cdots, i{+}j_{n}})^{\ast}{=}-\hat{H}^{\alpha}_{i{+}j_{1}, \cdots, i{+}j_{n}}\mathcal{C},
    \label{eq: condition}
\end{equation}
where the complex conjugation is defined with respect to the computational basis $|\vec{S}\rangle$. The proof of the theorem proceeds by using the mapping $\overline{i}{\equiv}i{+}L$, and expressing the Hamiltonian for the $2L$-site system with PBC as $H(2L){=}\sum_{\alpha,i{=}1}^{L} (\hat{H}^{\alpha}_{i{+}j_{1},{\cdots},i{+}j_{n}}{+}\hat{H}^{\alpha}_{\overline{i{+}j_{1}},{\cdots},\overline{i{+}j_{n}}})$, and then showing that for the state $|\Lambda\rangle_{\mathcal{C}}$, the action of the term $\hat{H}^{\alpha}_{i{+}j_{1}, {\cdots}, i{+}j_{n}}$ is canceled by that of $\hat{H}^{\alpha}_{\overline{i{+}j_{1}},{\cdots},\overline{i{+}j_{n}}}$ for all $i$ [see supplemental material (SM)~\cite{SM} for details], i.e.,
\begin{equation}
    \hat{H}^{\alpha}_{i{+}j_{1},\cdots,i{+}j_{n}}|\Lambda\rangle_{\mathcal{C}}{=}{-}\hat{H}^{\alpha}_{\overline{i{+}j_{1}},\cdots,\overline{i{+}j_{n}}}|\Lambda\rangle_{\mathcal{C}}.
    \label{eq: cancellation}
\end{equation}

When each $\hat{H}^{\alpha}_{i{+}j_{1}, {\cdots}, i{+}j_{n}}$ is imaginary, by noting that Eq.~\eqref{eq: condition} always holds for $\mathcal{C} {=} \mathbb{I}$, we obtain the following corollary to Theorem~\ref{theorem: 1}.

\begin{cor}
\label{corollary: 1}
The state 
\begin{equation}
    |\Lambda\rangle_{\mathbb{I}}{=} \sum_{|\vec{S}\rangle} |\vec{S}\rangle  {\otimes} |\vec{S}\rangle{=}\bigotimes_{i{=}1}^{L}\sum_{s_{i}{=}1}^{2S{+}1}|s_{i}{\rangle}|s_{\overline{i}}{\rangle}
    \label{eq: zero_energystate_for_imaginary_Hamiltonian}
\end{equation}
is always a zero-energy eigenstate of any Hamiltonian of the form Eq.~\eqref{eq: generic_Hamiltonian} defined on an even number of lattice sites ($2L$) where each $\hat{H}^{\alpha}_{i{+}j_{1}, {\cdots}, i{+}j_{n}}$ has a purely imaginary matrix representation in the computational basis. 
\end{cor}

Next, we show that the state $|\Lambda\rangle_{\mathcal{C}}$ is volume-law entangled for an arbitrary contiguous bipartition. By definition, the state $|\Lambda\rangle_{\mathcal{C}}$ is a product of ``dimers" between sites $i$ and $\overline{i}$. The entanglement between the spins on the sites $i$ and $\overline{i}$ is therefore strictly greater than $0$. It can be as high as $\text{log}(2S{+}1)$ depending upon the nature of $\mathcal{C}_{i}$ (the reduced density matrix of each dimer is just $\mathcal{C}_{i} \mathcal{C}_{i}^{\dagger}$, whose eigenvalues are just the singular values of $\mathcal{C}_{i}$), and is maximally entangled if and only if $\mathcal{C}_{i}$ is unitary. 

An arbitrary contiguous subsystem of size $N{\leq}L$ [shaded grey in Fig.~\ref{fig:1}$(a)$] cuts exactly $N$ bonds between the correlated spins leading to the entanglement entropy of the subsystem $\mathcal{S}_{N}{\propto} N$ thereby obeying volume-law scaling. Moreover, when $\mathcal{C}_{i}$ is unitary, the reduced density matrix over any contiguous subsystem is just proportional to the identity matrix, which is just the infinite-temperature Gibbs density matrix, hence the state is thermal for strictly local observables~\cite{ivanov2024volume, chiba2024exact}. However, for tailored subsystems chosen to consist solely of sites $\{i,\overline{i}\}$, the entanglement entropy is strictly zero, and non-local observables with support over these sub-systems exhibit highly athermal properties.     

We now illustrate the power of our general construction by illustrating exact excited zero-energy states for several integrable and non-integrable spin chains. We begin by considering spin-$1/2$ systems for which any Hamiltonian can be recast to the form given in Eq.~\eqref{eq: generic_Hamiltonian}. The spin operators at site $i$ can be represented by Pauli matrices $\sigma_{i}^{\mu} ({\mu}{=}x,y,z)$ and the local Hilbert space is defined by the $\sigma^{z}_{i}$-eigenvectors $|0\rangle_{i}$ and $|1\rangle_{i}$. In what follows, the state $|\Lambda\rangle_{\mathcal{C}}$ can be expressed in terms of Bell states between sites $i,j$ defined as
\begin{equation}
    |\Psi^{\pm}\rangle_{i, j}{=}\frac{(|00\rangle {\pm} |11\rangle)_{i,j}}{\sqrt{2}};~
|\Phi^{\pm}\rangle_{i,j} {=} \frac{(|01\rangle {\pm} |10\rangle)_{i,j}}{\sqrt{2}}.
\label{eq: bell states}
\end{equation}
For instance, consider the following non-integrable Hamiltonian discussed in Ref.~\cite{chiba2024exact}
\begin{equation}
\begin{split}
    H_{1}(L){=}\sum_{i{=}1}^{L}\big(J_{1}\sigma_{i}^{x}\sigma_{i{+}1}^{y}+J_{2}\sigma_{i}^{y}\sigma_{i+1}^{z}\big),
\end{split}
    \label{eq: spin-1/2 example2}
\end{equation}
with arbitrary values of the coupling constants $J_{1}$ and $J_{2}$. Here $n{=}2,j_{k}{=}k{-}1$ and each $\hat{H}^{\alpha}_{i, i{+}1}$ (i.e., $\sigma_{i}^{x}\sigma_{i{+}1}^{y}$ and $\sigma_{i}^{y}\sigma_{i+1}^{z}$) is purely imaginary in the computational basis. From Corollary~\ref{corollary: 1}, the state $|\Lambda\rangle_{\mathbb{I}}$ given in Eq.~\eqref{eq: zero_energystate_for_imaginary_Hamiltonian} is a zero-energy eigenstate of $H_{1}(2L)$ and takes a simple form $|\Lambda(S{=}1/2)\rangle_{\mathbb{I}}{=}\bigotimes_{i{=}1}^{L}|\Psi^{+}\rangle_{i,\overline{i}}$. This state is depicted in Fig.~\ref{fig:1}$(a)$ where all the dashed lines represent the Bell state $|\Psi^{+}\rangle$. Since the Hamiltonian $H_{1}(2L)$ has global spectral-reflection symmetry generated by the operator $\mathcal{M}_{1}{=}{\prod}_{i{=}1}^{2L}\sigma_{i}^{y}$, the zero-energy state mentioned above is a mid-spectrum excited state. Note that $|\Lambda(S{=}1/2)\rangle_{\mathbb{I}}$ is a zero-energy eigenstate of \textit{any} spin-$1/2$ Hamiltonian with even number of sites of form Eq.~\eqref{eq: generic_Hamiltonian} that has an \textit{odd} number of $\sigma_{y}$'s in each $\hat{H}^{\alpha}_{i{+}j_{1}, {\cdots}, i{+}j_{n}}$. 

Furthermore, when the half-chain length $L$ is a multiple of three, the operator $\mathcal{C}_{1}{=}{\prod}_{p{=}1}^{L/3}\sigma_{3p{-}2}^{x} \sigma_{3p{-}1}^{y}\sigma_{3p}^{z}$ satisfies the conditions given in Eq.~\eqref{eq: condition}. This allows us to identify another zero-energy eigenstate (and its translations) of $H_{1}(2L)$ of the form $|\Lambda\rangle_{\mathcal{C}_{1}}{=}\bigotimes_{i{=}1}^{L{-}2}|\Phi^{+}\rangle_{i,\overline{i}}|\Phi^{-}\rangle_{i{+}1,\overline{i{+}1}}|\Psi^{-}\rangle_{i{+}2,\overline{i{+}2}}$. The state $|\Lambda\rangle_{\mathcal{C}_{1}}$ is precisely the exact zero-energy eigenstate dubbed as Entangled Antipodal Pair (EAP) state presented in Ref.~\cite{chiba2024exact}. While Ref.~\cite{chiba2024exact} derives the Hamiltonian $H_{1}(2L)$ via an exhaustive search of all possible nearest-neighbor spin-$1/2$ interactions for which the predefined EAP state has zero-energy, we show that $H_1(2L)$ satisfies a very general sufficient condition for the existence of such an EAP state. Similarly, other zero energy eigenstates found in different spin-$1/2$ Hamiltonians of Ref~\cite{chiba2024exact, udupa2023weak} including the PXP chain~\cite{ivanov2024volume} can be understood using our construction (see SM~\cite{SM} for details with other examples).

Next, we investigate a spin-$S$ $XY$ chain, described by the Hamiltonian
\begin{equation}
    H_{XY}(L){=} \sum_{i=1}^{L} (J_{1}S_{i}^{x}S_{i+1}^{x}+J_{2}S_{i}^{y}S_{i+1}^{y}+hS_{i}^{z}),
    \label{eq: spin-S XY}
\end{equation}
where $J_{1},J_{2}$ and $h$ are arbitrary and $S_{i}^{\mu}$($\mu{=}x,y,z$) are spin${-}S$ operators at site $i$. The spin components define site parity operators $\mathcal{P}_{i}^{\mu}{=}e^{i\pi S_{i}^{\mu}}$. One can show that $\mathcal{P}_{i}^{\mu}$ commutes with $S_{i}^{\mu}$ but anti-commutes with $S_{i}^{\mu'} (\mu{\neq}\mu')$. Hence when $L$ is even, in the orthonormal basis $|m\rangle_{i}$ that are eigenstates of $S_{i}^{z}$ with eigenvalues $m{=}{-}S, {-}S{+}1, {\cdots}, S$, the operator defined by $\mathcal{C}_{XY}{=}\prod_{p{=}1}^{L/2}\mathcal{P}_{2p{-}1}^{y}\mathcal{P}_{2p}^{x}$, satisfies Eq.~\eqref{eq: condition}. As a result of Theorem~\ref{theorem: 1}, the Hamiltonian $H_{XY}(2L)$ hosts an exact zero-energy eigenstate $|\Lambda\rangle_{\mathcal{C}_{XY}}$ (and its translations). Using the action of the on-site parity operators on the basis states
$\mathcal{P}_{i}^{x}|m\rangle_{i}{=}({-}1)^{S}|{-}m\rangle_{i}$ and $\mathcal{P}_{i}^{y}|m\rangle_{i}{=}({-}1)^{S{+}m}|{-}m\rangle_{i}$, we can express the state as
\begin{equation}
\begin{split}
&|\Lambda\rangle_{\mathcal{C}_{XY}}=\bigotimes_{i{=}1}^{L{-}1}|\kappa_{S}^{-}\rangle_{i,\overline{i}}|\kappa_{S}^{+}\rangle_{i{+}1,\overline{i{+}1}},\hspace{0.4cm}\text{where}\\
  &|\kappa_{S}^{\pm}\rangle_{i,j}{=}\frac{1}{\sqrt{2S{+}1}}\sum_{p{=}0}^{2S}({\pm}1)^{p}(S^{-})^{p}|S\rangle_{i}{\otimes}(S^{+})^{p}|{-}S\rangle_{j}.
\end{split}
\label{eq: Lambda_S_XY state}
\end{equation}
For $S{>}1/2$, $H_{XY}$ is non-integrable~\cite{schecter2019weak} for generic values of the couplings, while for $S{=}1/2$, $H_{XY}$ is integrable (reduces to the transverse field Ising model for $J_{2}{=}0$)~\cite{Franchini_XY_17}. For $S{=}1/2$, this state reads $|\Lambda\rangle_{\mathcal{C}_{XY}}{=}\bigotimes_{i{=}1}^{L{-}1}|\Phi^{-}\rangle_{i,\overline{i}}|\Phi^{+}\rangle_{i{+}1,\overline{i{+}1}}$ and shows that volume-law eigenstates exist even in integrable models. 

The sufficient conditions we derive allow us to construct a large family of Hamiltonians that will inherently host the specific volume-law entangled eigenstate of the particular form of Eq.~\eqref{eq: zero_energystate_1d}. For instance, the zero-energy state $|\Lambda\rangle_{\mathcal{C}_{XY}}$ described in Eq.~\eqref{eq: Lambda_S_XY state} is also a zero-energy eigenstate of the spin-$S$ Kitaev chain~\cite{Kitaev2006anyons, sen2010spin, mohapatra2023pronounced} given by the Hamiltonian $H_{K}(2L){=}\sum_{i{=}1}^{L}(J_{1}S_{2i{-}1}^xS_{2i}^x {+}J_{2} S_{2i}^y S_{2i{+}1}^y)$ for arbitrary values of $J_{1}$ and $J_{2}$. Furthermore, Hamiltonians involving longer-range (beyond nearest neighbors) interactions between spins, e.g. $H_{\rm NNNN}(2L){=}\sum_{i{=}1}^{2L} (J_{1}S_{i}^{x}S_{i{+}3}^{x}{+}J_{2}S_{i}^{y}S_{i{+}3}^{y})$ or models with more than two-spin interactions, e.g. $H_{3{-}4}(2L){=}\sum_{i{=}1}^{2L} (J_{1}S_{i}^{z}S_{i{+}1}^{z}S_{i{+}2}^{z}{+}J_{2}S_{i}^{x}S_{i{+}1}^{z}S_{i{+}2}^{y}S_{i{+}3}^{z})$
host the same zero-energy state $|\Lambda\rangle_{\mathcal{C}_{XY}}$ in their spectrum as they satisfy Eq.~\eqref{eq: condition}. Many such Hamiltonians exist physically or can be engineered experimentally.

\textbf{\textit{Generalization to arbitrary dimensions.}}
To investigate higher-dimensional or amorphous systems, it is convenient to revisit the one-dimensional case and rearrange the $2L$ sites labeled by indices ${1}, \overline{1}, {\cdots},{L}, \overline{L}$ into the configuration shown in Fig.~\ref{fig:1}$(b)$. Our goal is to determine local Hamiltonians that host the state $|\Lambda\rangle_\mathcal{C}$ defined in Eq.~\eqref{eq: zero_energystate_1d} as a zero energy eigenstate provided the operator $\mathcal{C}$ satisfies conditions analogous to those in Eq.~\eqref{eq: condition}. 

Similar to the one-dimensional case, for all the interaction terms $\hat{H}_{j_{1},{\cdots}, j_{n}}^{\alpha}$ between any $n{-}$sites $j_{1},{\cdots}, j_{n}$ (where $j_k{\in}\{1,\overline{1},{\cdots}, L,\overline{L}\}$ and $j_{k}{\neq} \overline{j}_{k'}$ for all $k,k'$), restricted to the form $\hat{H}_{j_{1},{\cdots}, j_{n}}^{\alpha}{=}\prod_{k{=}1}^{n}(S_{j_{k}}^{\mu_{k}})^{q_{k}}$ and satisfying $ \{\hat{H}^{\alpha}_{j_{1},{\cdots},j_{n}},\mathcal{C}\mathcal{K}\}{=}0$ (with $\mathcal{C}_{i}{=}\mathcal{C}_{\overline{i}}$), we find that
\begin{equation}
    (\hat{H}^{\alpha}_{j_{1},\cdots,j_{n}} + \hat{H}^{\alpha}_{\overline{j}_{1},\cdots,\overline{j}_{n}})|\Lambda\rangle_{\mathcal{C}}=0,
    \label{eq: cancelation_in_arb_dimension}
\end{equation}
where we have considered the identification $\overline{\overline{j}}_{k}{=}j_{k}$. Consequently, all Hamiltonians of the form $H{=}\sum_{\alpha,\{j\}} J_{\alpha}(\hat{H}^{\alpha}_{j_{1},{\cdots},j_{n}} {+} \hat{H}^{\alpha}_{\overline{j}_{1},{\cdots},\overline{j}_{n}})$ with arbitrary coefficients $J_{\alpha}$, host $|\Lambda\rangle_{\mathcal{C}}$ as a zero energy eigenstate.

For instance, consider the Hamiltonian of the nearest-neighbor spin-$S$ $XY$ model with a transverse field on a $M{\times}L$ rectangular lattice with PBC in both directions,
\begin{equation}
\begin{split}
    H_{XY}^{2D}(M{\times}L){=}\sum_{\eta,i,j}J_{\eta}(S_{r_{i,j}}^{\eta}S_{r_{i{+}1,j}}^{\eta}{+}S_{r_{i,j}}^{\eta}S_{r_{i,j{+}1}}^{\eta}){+}h\sum_{i,j}S_{r_{i,j}}^{z},    
\end{split}
\label{eq: 2D XY model Hamiltonian}
\end{equation}
where $S_{r_{i,j}}^{\eta}(\eta{\in}\{x,y\})$ and $S_{r_{i,j}}^{z}$ are the spin operators at site $(i,j)$ indexed by $r_{i,j}{=}i{+}M{\times}(j{-}1)$ with $i{\in}\{1, {\cdots},M\}$  and $j{\in}\{1, {\cdots},L\}$ labeling the horizontal and vertical directions, respectively, with $(1,1)$ at the bottom-left corner, and $J_{\eta}$ and $h$ are arbitrary couplings.

\begin{figure}
    \centering
    \includegraphics[scale=0.35]{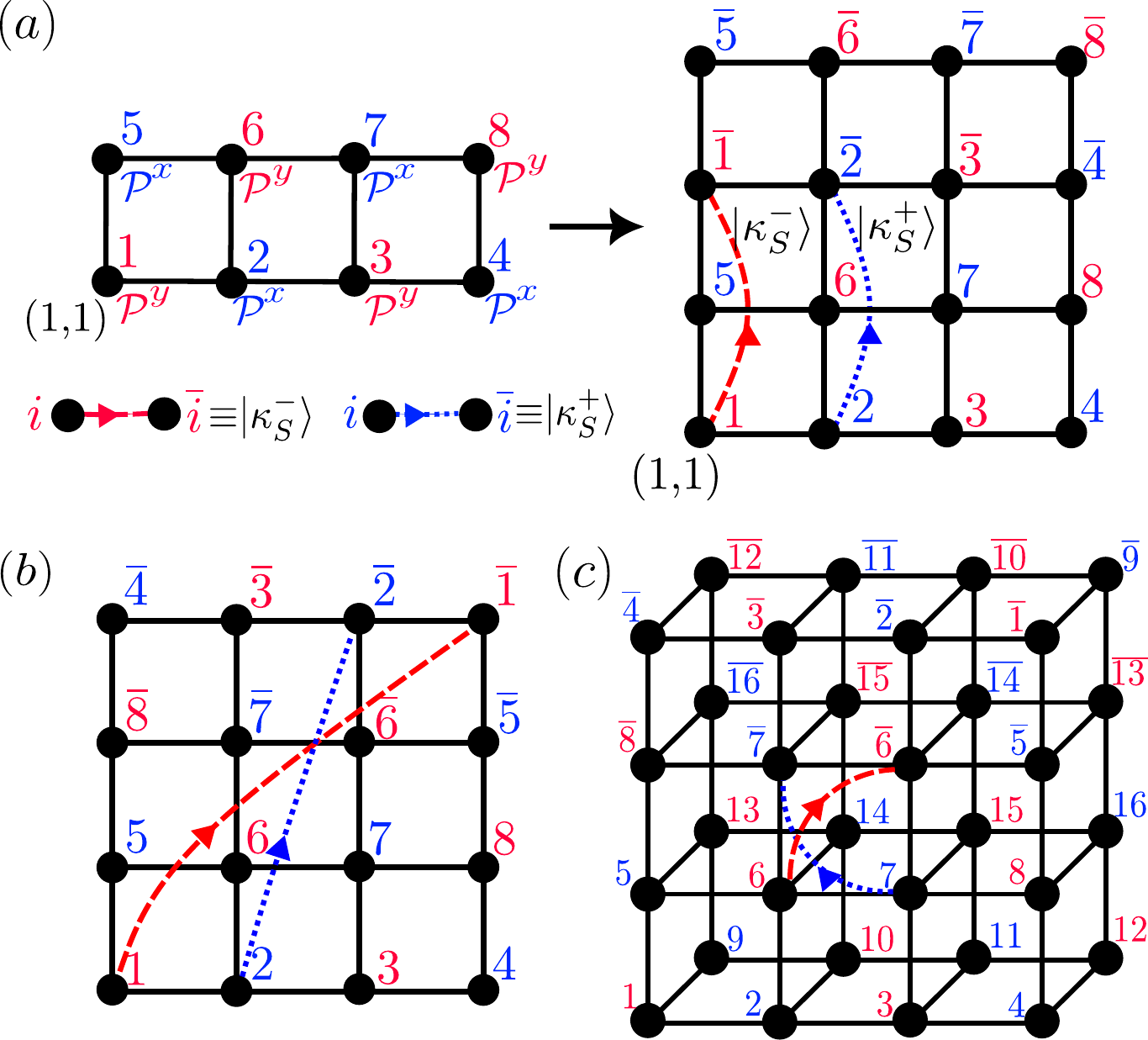}
    \caption{Schematic representation of the zero energy eigenstate $|\Lambda\rangle_\mathcal{C}$ for spin-$S$ $XY$ model with transverse field in different geometries. The black dots denote the sites and the solid lines between them indicate the $ XY$ interaction. The spins at red (blue) sites $r$ are entangled only with their partners $\overline{r}$ in the state $|\kappa_{S}^{-}\rangle$ ($|\kappa_{S}^{+}\rangle$) [defined in Eq.~\eqref{eq: Lambda_S_XY state}]. $(a)$ The $\mathcal{C}_{XY}^{2D}$ that satisfies Eq.~\eqref{eq: condition} for a $M{\times}L{=}4{\times}2$ rectangular lattice with PBC (left part), which gives a zero energy state for $M{\times}2L$ lattice with PBC (right part). $(b)$ Another distinct zero energy state on a $M{\times}2L$ lattice with PBC with a different pairing pattern. $(c)$ Eigenstate on a $M{\times}2L{\times}N$ cuboid, with either PBC or open boundaries, with $(M,L,N){=}(4,2,2).$}
    \label{fig:2}
\end{figure}

Now using the mapping $\overline{r_{i,j}}{=}r_{i,j{+}L}$, the Hamiltonian for a $M{\times}2L$ lattice with PBC can be expressed as
\begin{equation}
\begin{split}
    H&_{XY}^{2D}(M{\times}2L){=}\sum_{\eta,i,j}\Big[J_{\eta}(S_{r_{i,j}}^{\eta}S_{r_{i{+}1,j}}^{\eta}{+}S_{\overline{r_{i,j}}}^{\eta}S_{\overline{r_{i{+}1,j}}}^{\eta})\\
    &{+}J_{\eta}(S_{r_{i,j}}^{\eta}S_{r_{i,j{+}1}}^{\eta}{+}S_{\overline{r_{i,j}}}^{\eta}S_{\overline{r_{i,j{+}1}}}^{\eta}){+}h(S_{r_{i,j}}^{z}{+}S_{\overline{r_{i,j}}}^{z})\Big].    
\end{split}
\label{eq: 2D XY model Hamiltonian_relabelled}
\end{equation}
In the computational basis $|m\rangle_{r_{i,j}}$ that are eigenstates of $S_{r_{i,j}}^{z}$, each $\hat{H}_{\{r\}}^{\alpha}$ of the Hamiltonian, i.e., $S_{r_{i,j}}^{x}S_{r_{i{+}1,j}}^{x},S_{r_{i,j}}^{y}S_{r_{i{+}1,j}}^{y}, S_{r_{i,j}}^{x}S_{r_{i,j{+}1}}^{x}$, $S_{r_{i,j}}^{y}S_{r_{i,j{+}1}}^{y}$ and $S_{r_{i,j}}^{z}$ are real and for even values of $M$ and $L$, anticommute with the operator [shown in the left part of Fig.~\ref{fig:2}$(a)$]
\begin{equation}
    \mathcal{C}_{XY}^{2D}{=}\prod_{i,j}(\mathcal{P}_{r_{i,j}}^{y})^{\delta_{(i{+}j){\rm~mod~}2,0}}(\mathcal{P}_{r_{i,j}}^{x})^{\delta_{(i{+}j){\rm~mod~}2,1}}
    \label{eq: operator_2DXY}
\end{equation}
[hence satisfy Eq.~\eqref{eq: condition}].
Thus, using a generalization of Theorem~\ref{theorem: 1} (see SM~\cite{SM}) and Eq.~\eqref{eq: cancelation_in_arb_dimension}, the Hamiltonian $H_{XY}^{2D}(M{\times}2L)$ hosts a zero energy eigenstate of the form $|\Lambda\rangle_{C_{XY}}^{2D}$ which can be expressed as 
\begin{equation}
    |\Lambda\rangle_{\mathcal{C}_{XY}}^{2D}{=}\bigotimes_{i,j}(|\kappa_{S}^{-}\rangle_{r,\overline{r}})^{\delta_{(i{+}j){\rm~mod~}2,0}}
    (|\kappa_{S}^{+}\rangle_{r,\overline{r}})^{\delta_{(i{+}j){\rm~mod~}2,1}}, 
    \label{eq: zeroenergy_eigstate_2DXY}
\end{equation}
where $|\kappa_S^{\pm}\rangle$ are defined in Eq.~\eqref{eq: Lambda_S_XY state}. The state $|\Lambda\rangle_{\mathcal{C}_{XY}}^{2D}$ is shown in the right part of Fig.~\ref{fig:2}$(a)$, where the spins at sites $r_{i,j}$ and $\overline{r_{i,j}}$ are in the state $|\kappa_{S}^{-}\rangle$ when $i{+}j$ is even (marked in red) and in the state $|\kappa_{S}^{+}\rangle$ when $i{+}j$ is odd (marked in blue). 
Similarly, with some other choices of labeling of the sites of the $M{\times}2L$ lattice (e.g., see Fig.~\ref{fig:2}$(b)$ where $\overline{r_{i,j}}{=}r_{M{+}1{-}i,2L{+}1{-}j}$), the Hamiltonian can be cast to the form given in Eq.~\eqref{eq: 2D XY model Hamiltonian_relabelled} while the operator $\mathcal{C}_{XY}^{2D}$ given in Eq.~\eqref{eq: operator_2DXY} satisfy the required condition Eq.~\eqref{eq: condition}, thus we can construct many more zero-energy eigenstates of the kind $|\Lambda\rangle_{\mathcal{C}_{XY}}^{2D}$ for $H_{XY}^{2D}$ with different pairing patterns. Notably, the number of such zero energy eigenstates scales with the volume of the system $\mathcal{O}(ML)$ (see SM~\cite{SM} for details). The same strategy can be used to find zero-energy eigenstates of spin-$S$ $XY$ model in different geometries (e.g., see Fig.~\ref{fig:2}$(c)$ and more in the SM~\cite{SM}). Using our construction, many such zero-energy eigenstates can be found for other spin Hamiltonians in arbitrary dimensions. Unlike the one-dimensional case, in higher dimensions, there can be special contiguous bipartitions for which the entanglement is not a volume-law, although it is for most of them~\cite{SM}.

\textbf{\textit{Discussion.}} 
In this work, we have introduced a general method for constructing a specific class of exact zero-energy excited eigenstates in spin-$S$ models that satisfy certain conditions. This construction only uses the existence of an operator, which, along with complex conjugation, anti-commutes with each term of the Hamiltonian.
The existence of such an operator is easy to guess and check, and this directly allows us to construct exact eigenstates without analyzing the complex microscopic details of the model. We demonstrated the broad applicability of our construction by finding new exact eigenstates for various natural examples of spin-$S$ Hamiltonians. We also recovered some previously known results for spin-$1/2$ models~\cite{ivanov2024volume, chiba2024exact} as special cases of our construction. 

In one dimension, the states we identified exhibit volume-law scaling of entanglement for general contiguous bipartitions. When the operator $\mathcal{C}$ is unitary, they exhibit infinite-temperature thermal properties for all strictly local observables. This phenomenology mostly extends to higher dimensions, although a few contiguous bipartitions might not display a volume-law scaling of entanglement (see SM~\cite{SM} for details). However, in any dimension, these states are athermal for non-local but few-body observables and hence can be considered examples of highly entangled QMBS. Indeed, by construction, they are simultaneous eigenstates of non-commuting local operators, consistent with the proposed definition of QMBS based on commutant algebras~\cite{moudgalya2022exhaustive}. It would be interesting to see whether other examples of QMBS exhibiting superficially similar structures, like the exponentially many exact eigenstates in the 2D PXP model~\cite{lin2020quantum}, fit into this construction. 

Building on this foundation, it is natural to wonder if extensions of these ideas to models beyond spins such as fermionic and bosonic systems would work, or generalizations such as going beyond pairwise entanglement to entangled clusters of $l$-spins, with $l{\geq}3$ can be engineered. Along with the theoretical interest in states close to completely satisfying ETH, the existence of such states can also aid in the experimental measurement of Out-of-Time-Ordered Correlators (OTOC), as discussed in \cite{ivanov2024volume}.

Finally, on a different note, while we focused on constructing non-integrable models with these states, there are also integrable models, such as the one-dimensional transverse-field Ising model, which host such states. While these states are athermal and violate the ETH~\cite{srednicki1994chaos} in non-integrable models (at least with respect to some non-local but few-body observables), it is worth exploring the implications of their existence in integrable models to the generalized ETH~\cite{cassidy2011generalized,vidmar2016generalized, Ilievski2015complete, Mestyan_2014_Generalized_ETH, Pozsgay_2017_Generalized_ETH}, the natural analog of ETH for integrable models.

\textbf{\textit{Acknowledgments}.} We thank Andrew Ivanov, Bal\'azs Pozsgay, Maurizio Fagotti, and Diptiman Sen for useful discussions. Computational portions of this work were carried out on the Nandadevi supercomputer maintained and supported by the Institute of Mathematical Sciences' High-Performance Computing Center, India. This research was supported by the Munich Center for Quantum Science and Technology (MCQST), and in part by the International Centre for Theoretical Sciences (ICTS) for participating in the program - Stability of Quantum Matter in and out of Equilibrium at Various Scales (code: ICTS/SQMVS2024/01). S. Moudgalya acknowledges the hospitality of the Institute of Mathematical Sciences (IMSc), Chennai, where this collaboration was initiated. S. Mohapatra is grateful to the School of Natural Sciences, Technical University of Munich (TUM) for its hospitality in the summer of 2024 where a part of this work was completed.

\bibliography{references}

\newpage 
\cleardoublepage

\onecolumngrid
\begin{center}
\textbf{\large Supplemental Material for ``Exact volume-law entangled zero-energy eigenstates in a large class of spin models"}\\[5pt]

\begin{center}
 {\small Sashikanta Mohapatra$^{1,2}$, Sanjay Moudgalya$^{3,4}$ and Ajit C. Balram$^{1,2}$}  
\end{center}

\begin{center}
{\sl \footnotesize
$^{1}$Institute of Mathematical Sciences, CIT Campus, Chennai 600113, India

$^{2}$Homi Bhabha National Institute, Training School Complex, Anushaktinagar, Mumbai 400094, India
    
$^{3}$School of Natural Sciences, Technische Universit\"{a}t M\"{u}nchen (TUM), James-Franck-Str. 1, 85748 Garching, Germany

$^{4}$Munich Center for Quantum Science and Technology (MCQST), Schellingstr. 4, 80799 M\"{u}nchen, Germany
}
\end{center}

\begin{quote}
{\small This supplemental material contains some technical details and additional examples. In Sec.~\ref{sec: Proof of Theorem 1}, we prove Theorem~\ref{theorem: 1} of the main text. In Sec.~\ref{sec: generalization_and_new_examples}, we present the generalization of Theorem~\ref{theorem: 1} to higher dimensions and applications to the nearest neighbor spin-$S$ $XY$ model on a hypercube that we considered in the main text. Finally, in Sec.~\ref{sec: previous_examples}, we reproduce some zero energy eigenstates of previously studied Hamiltonians~\cite{udupa2023weak,ivanov2024volume,chiba2024exact} using our framework. }
\end{quote}
\end{center}

\vspace*{0.4cm}

\setcounter{equation}{0}
\setcounter{theorem}{0}
\setcounter{figure}{0}
\setcounter{table}{0}
\setcounter{page}{1}
\setcounter{section}{0}
\makeatletter
\renewcommand{\theequation}{S\arabic{equation}}
\renewcommand{\thetheorem}{S\arabic{theorem}}
\renewcommand{\thefigure}{S\arabic{figure}}
\renewcommand{\thesection}{S\Roman{section}}
\renewcommand{\thepage}{\arabic{page}}
\renewcommand{\thetable}{S\arabic{table}}

\section{Proof of Theorem 1 of main text: $|\Lambda\rangle_\mathcal{C}$ is a zero energy eigenstate of $H(2L)$}
\label{sec: Proof of Theorem 1}

For convenience, let us restate the Hamiltonian and Theorem 1 from the main text here. We are considering the class of spin-$S$ chains of size $L$ with a translation invariant interaction [with periodic boundary condition (PBC)] of the form
\begin{equation}
    H(L){=}\sum_{\alpha,i{=}1}^{L} \hat{H}^{\alpha}_{i{+}j_{1},\cdots,i{+}j_{n}}{=}\sum_{\alpha,i{=}1}^{L}J_{\alpha}\prod_{k{=}1}^{n} (S_{i{+}j_{k}}^{\mu_k})^{q_{k}},
    \label{SMeq: generic_Hamiltonian}
\end{equation}
where $S_{i{+}j_{k}}^{\mu_{k}}(\mu_{k}{ \in }\{x,y,z\})$ is the spin-$S$ operator at site $i{+}j_{k}$ and $q_{k}$ is a non-negative integer. The term $\hat{H}^{\alpha}_{i+j_{1},{\cdots}, i{+}j_{n}}$ represents the $\alpha$-type (determined by $\{\mu_{k}\}$ and $\{q_{k}\}$) of interaction between $n{-}$spins at the sites $\{i{+}j_{1},{\cdots}, i{+}j_{n}\}$ ($0{\leq}j_{k}{<}L$), with coupling strength $J_{\alpha}$. The Hamiltonian is defined in the Hilbert space $\mathcal{H}_{L}$ of dimension $\mathcal{D}_{L}{=}(2S{+}1)^{L}$ on the global computational basis $|\vec{S}\rangle{=}{\bigotimes}_{i{=}1}^{L}|s_{i}\rangle$ $(\mathcal{H}_{L}{\equiv}\text{span}\{|\vec{S}\rangle\})$, where $|s_{i}\rangle$ denotes the $(2S{+}1)$ dimensional local Hilbert space on site $i$. For this setting, we stated the following theorem in the main text:

\begin{thm}
\label{SMtheorem: 1}
If each $\hat{H}^{\alpha}_{i{+}j_{1},{\cdots}, i{+}j_{n}}$ in the Hamiltonian $H(L)$ anti-commutes with an operator $\mathcal{C}\mathcal{K}$, where $\mathcal{C}{=}\prod_{i{=}1}^{L}\mathcal{C}_{i}$ with $\mathcal{C}_{i}$ on-site invertible operators, and $\mathcal{K}$ is the complex conjugation operator, i.e.,  
\begin{equation}
    \{\hat{H}^{\alpha}_{i+j_{1},\cdots,i+j_{n}},\mathcal{C}\mathcal{K}\}{=}0 \hspace{0.15cm} \text{for}\hspace{0.15cm} i{\in}\{1,{\cdots}, L\},
    \label{SMeq: anti_commuation}
\end{equation}
then the state
\begin{equation}
\label{SMeq: zero_energystate_1d}
\begin{split}
    |\Lambda\rangle_{\mathcal{C}}{=}\sum_{|\vec{S}\rangle\in\mathcal{H}_{L}} \mathcal{C}|\vec{S}\rangle_{1,\cdots,L}  \otimes |\vec{S}\rangle_{\overline{1},\cdots,\overline{L}}
    {=}\bigotimes_{i{=}1}^{L}\sum_{s{=}1}^{2S{+}1}\mathcal{C}_{i}|s_{i}{\rangle}|s_{\overline{i}}{\rangle},
\end{split}
\end{equation}
where $\overline{i}{\equiv} i{+}L$, is an exact zero energy eigenstate of the Hamiltonian $H(2L)$ defined in Eq.~\eqref{SMeq: generic_Hamiltonian} on a system of size $2L$ with PBC.
\end{thm}

\begin{proof}
Throughout this proof, we use the identification $\overline{i}{\equiv}i{+}L$. We can then express the Hamiltonian for the $2L$-site system with PBC as
\begin{equation}
   H(2L){=}\sum_{\alpha,i{=}1}^{L} (\hat{H}^{\alpha}_{i{+}j_{1},\cdots,i{+}j_{n}}{+}\hat{H}^{\alpha}_{\overline{i{+}j_{1}},\cdots,\overline{i{+}j_{n}}}) 
\end{equation}
and then the proof proceeds by showing that the action of the term $\hat{H}^{\alpha}_{i{+}j_{1}, \cdots, i{+}j_{n}}$ on the state $|\Lambda\rangle_{\mathcal{C}}$ is exactly canceled by that of $\hat{H}^{\alpha}_{\overline{i{+}j_{1}},\cdots,\overline{i{+}j_{n}}}$ on the same state for all $i$.

This is straightforward to show when each $\hat{H}^{\alpha}_{i{+}j_{1}, {\cdots}, i{+}j_{n}}$ acts entirely on either the first $L$ sites ($1 {\leq} i{+}j_{1}{\leq}L$ and $1 {\leq} i{+}j_{n}{\leq}L$), or the last $L$ sites ($\overline{1}{\leq}i{+}j_{1}{\leq}\overline{L}$ and $\overline{1}{\leq}i{+}j_{n}{\leq}\overline{L}$). By the anti-commutation relation in Eq.~\eqref{SMeq: anti_commuation}, we simply mean
\begin{equation}
    \mathcal{C}(\hat{H}^{\alpha}_{i{+}j_{1}, \cdots, i{+}j_{n}})^{\ast}{=}-\hat{H}^{\alpha}_{i{+}j_{1}, \cdots, i{+}j_{n}}\mathcal{C},
    \label{SMeq: condition}
\end{equation}
where the complex conjugation is defined with respect to the computational basis $|\vec{S}\rangle$. Hence, we obtain
\begin{subequations} 
\begin{align} 
&\hat{H}^{\alpha}_{i{+}j_{1},\cdots,i{+}j_{n}}|\Lambda\rangle_{\mathcal{C}}{=}{-}\sum_{|\vec{S}\rangle} \mathcal{C}(\hat{H}^{\alpha}_{i{+}j_{1},\cdots,i{+}j_{n}})^{\ast} |\vec{S}\rangle{\otimes} |\vec{S}\rangle
\label{SMeq: action_of_left_part_on_Lambda}\\
&\hat{H}^{\alpha}_{\overline{i{+}j_{1}},\cdots,\overline{i{+}j_{n}}}|\Lambda\rangle_{\mathcal{C}}{=}{\sum_{|\vec{S}\rangle}}\mathcal{C}|\vec{S}\rangle{\otimes}  \hat{H}^{\alpha}_{\overline{i{+}j_{1}},\cdots,\overline{i{+}j_{n}}}|\vec{S}\rangle.
\label{SMeq: action_of_right_part_on_Lambda}
\end{align} 
\end{subequations}
Next, in Eqs.~\eqref{SMeq: action_of_left_part_on_Lambda} and \eqref{SMeq: action_of_right_part_on_Lambda} we insert the resolution of the identity operator $\mathbb{I}_{1,{\cdots}, L}{=}(\sum_{\vec{S'}}|\vec{S'}\rangle\langle \vec{S'}|)_{1,\cdots, L}$ and $\mathbb{I}_{\overline{1},{\cdots}, \overline{L}}{=}(\sum_{\vec{S'}}|\vec{S'}\rangle\langle \vec{S'}|)_{\overline{1},{\cdots}, \overline{L}}$ respectively to get

\begin{subequations} 
\begin{align} 
&\hat{H}^{\alpha}_{i{+}j_{1},\cdots,i{+}j_{n}}|\Lambda\rangle_{\mathcal{C}}{=}{-}\sum_{|\vec{S}\rangle,|\vec{S'}\rangle} f^{\ast}(S',S)\  \mathcal{C}|\vec{S'}\rangle\otimes|\vec{S}\rangle,
\label{SMeq: action_of_left_part_on_Lambda_2}\\
&\hat{H}^{\alpha}_{\overline{i{+}j_{1}},\cdots,\overline{i{+}j_{n}}}|\Lambda\rangle_{\mathcal{C}}{=}{\sum_{|\vec{S}\rangle,|\vec{S'}\rangle}}f(S',S)\   \mathcal{C}|\vec{S}\rangle\otimes |\vec{S'}\rangle,
\label{SMeq: action_of_right_part_on_Lambda_2}
\end{align}
\end{subequations}
where the function $f(S',S){=}\langle \vec{S'}| (\hat{H}^{\alpha}_{i{+}j_{1},\cdots,i{+}j_{n}}) |\vec{S}\rangle$ with $f^{\ast}(S',S){=}\langle \vec{S'}| (\hat{H}^{\alpha}_{i{+}j_{1},\cdots,i{+}j_{n}})^{\ast} |\vec{S}\rangle$. Using the hermiticity of $\hat{H}^{\alpha}_{i{+}j_{1},\cdots,i{+}j_{n}} $ in this basis, i.e., $f^{\ast}(S',S){=}f(S,S')$, we interchange $|\vec{S}\rangle$ and $|\vec{S'}\rangle$ in Eq.~\eqref{SMeq: action_of_left_part_on_Lambda_2} to get 
\begin{equation}   \hat{H}^{\alpha}_{i{+}j_{1},\cdots,i{+}j_{n}}|\Lambda\rangle_{\mathcal{C}}{=}{-}\hat{H}^{\alpha}_{\overline{i{+}j_{1}},\cdots,\overline{i{+}j_{n}}}|\Lambda\rangle_{\mathcal{C}}\hspace{0.2cm}\text{for}\hspace{0.2cm} i{+}j_{n}{\leq}L.
\end{equation}
For terms that straddle the first $L$ sites and last $L$ sites, i.e., when $i{+}j_{1}{\leq}L$ and $i{+}j_{n}{>}L$ (or  $i{+}j_{1}{\leq}\overline{L}$ and $i{+}j_{n}{<}L$), given the form of the Hamiltonian in Eq.~\eqref{SMeq: generic_Hamiltonian}, we can decompose the terms $\hat{H}^{\alpha}_{i{+}j_{1},\cdots,i{+}j_{n}}$ and $\hat{H}^{\alpha}_{\overline{i{+}j_{1}},\cdots,\overline{i{+}j_{n}}}$ as
\begin{equation}
    \hat{H}^{\alpha}_{i{+}j_{1},\cdots,i{+}j_{n}}{=}\mathcal{A}^{\alpha}_{i{+}j_{1}, \cdots, L}{\otimes}\mathcal{B}^{\alpha}_{\overline{1}, \cdots, \overline{i{+}j_{n}{-}L}}~~\text{and}~~\hat{H}^{\alpha}_{\overline{i{+}j_{1}},\cdots,\overline{i{+}j_{n}}}{=}\mathcal{B}^{\alpha}_{1, \cdots, i{+}j_{n}{-}L}{\otimes}\mathcal{A}^{\alpha}_{\overline{i{+}j_{1}}, \cdots, \overline{L}},
    \label{SMeq: partition of H into A and B}
\end{equation}
where $\mathcal{A}^{\alpha}_{i{+}j_{1}, \cdots, L}{=}\prod_{k{=}1}^{i{+}j_{k}{=}L} (S_{i{+}j_{k}}^{\mu_{k}})^{q_{k}}$ and $\mathcal{B}^{\alpha}_{1, \cdots, i{+}j_{n}{-}L}{=}\prod_{i{+}j_{k}{=}1}^{i{+}j_{n}{-}L} (S_{i{+}j_{k}}^{\mu_{k}})^{q_{k}}$ acts on the first $L$ sites while $\mathcal{A}^{\alpha}_{\overline{i{+}j_{1}}, \cdots, \overline{L}}{=}$ $\prod_{k{=}1}^{i{+}j_{k}{=}L} (S_{\overline{i{+}j_{k}}}^{\mu_{k}})^{q_{k}}$ and $\mathcal{B}^{\alpha}_{\overline{1}, \cdots, \overline{i{+}j_{n}{-}L}}{=}\prod_{i{+}j_{k}{=}1}^{i{+}j_{n}{-}L} (S_{\overline{i{+}j_{k}}}^{\mu_{k}})^{q_{k}}$ acts on the last $L$ sites. Since both $\hat{H}^{\alpha}_{i{+}j_{1},\cdots, i{+}j_{n}}$ and the operator $\mC$ are just products of on-site operators, the condition stated in Eq.~\eqref{SMeq: condition} implies that for each of the spins we should have 
\begin{equation}
     \mC_{i{+}j_{k}}[(S_{i{+}j_{k}}^{\mu_{k}})^{q_{k}}]^{\ast}=K_{i{+}j_{k}}(S_{i{+}j_{k}}^{\mu_{k}})^{q_{k}}\mC_{i{+}j_{k}}~~\text{with}~~\prod_{k{=}1}^{n}K_{i{+}j_{k}}{=}{-}1,
     \label{SMeq: on_site_condition}
\end{equation}
where $K_{i{+}j_{k}}$ is a non-zero complex number. Thus for terms of the form given in Eq.~\eqref{SMeq: partition of H into A and B}, we should have
\begin{equation}
    \mathcal{C}(\mathcal{A}^{\alpha}_{i{+}j_{1}, \cdots, L})^{\ast}{=}K_{\rm{left}}\mathcal{A}^{\alpha}_{i{+}j_{1}, {\cdots}, L}\mathcal{C}~~\text{and}~~
    \mathcal{C}(\mathcal{B}^{\alpha}_{\overline{1}, \cdots, \overline{i{+}j_{n}{-}L}})^{\ast}{=}K_{\rm{right}}\mathcal{B}^{\alpha}_{\overline{1}, \cdots, \overline{i{+}j_{n}{-}L}}\mathcal{C}~~\text{with}~~K_{\rm{left}}{\times}K_{\rm{right}}{=}{-}1,
    \label{SMeq: splitting_condition}
\end{equation}
where $K_{\rm{left}}{=}\prod_{l{=}i{+}j_{1}}^{L}K_{l}$ and $K_{\rm{right}}{=}\prod_{l{=}L{+}1}^{i{+}j_{n}}K_{l}$. As we show in Lemma~\ref{Lemma1} below, the only non-zero value of $K_{l}$ can be ${\pm}1$. From Eq.~\eqref{SMeq: splitting_condition} we then get either $K_{\rm{left}}{=}1$ and $K_{\rm{right}}{=}{-}1$ or $K_{\rm{left}}{=}{-}1$ and $K_{\rm{right}}{=}1$, i.e.,
\begin{equation}
    \mathcal{C}(\mathcal{A}^{\alpha}_{i{+}j_{1}, \cdots, L})^{\ast}{=}{\pm}\mathcal{A}^{\alpha}_{i{+}j_{1}, {\cdots}, L}\mathcal{C},\;\;
    \mathcal{C}(\mathcal{B}^{\alpha}_{\overline{1}, \cdots, \overline{i{+}j_{n}{-}L}})^{\ast}{=}{\mp}\mathcal{B}^{\alpha}_{\overline{1}, \cdots, \overline{i{+}j_{n}{-}L}}\mathcal{C}.
\label{SMeq: spilting_condition_@}
\end{equation}
Using the decompositions shown in Eq.~(\ref{SMeq: splitting_condition}), and then using Eq.~(\ref{SMeq: spilting_condition_@}), we obtain
\begin{subequations}
\begin{align}
    &\hat{H}^{\alpha}_{i{+}j_{1},\cdots,i{+}j_{n}}|\Lambda\rangle_{\mathcal{C}}{=}\sum_{|\vec{S}\rangle}\mathcal{A}^{\alpha}_{i{+}j_{1}, \cdots, L}\mathcal{C}|\vec{S}\rangle
    {\otimes}   \mathcal{B}^{\alpha}_{\overline{1}, \cdots, \overline{i{+}j_{n}{-}L}}|\vec{S}\rangle{=}{\pm}\sum_{|\vec{S}\rangle}\mathcal{C}(\mathcal{A}^{\alpha}_{i{+}j_{1}, \cdots, L})^{\ast}|\vec{S}\rangle
    {\otimes}   \mathcal{B}^{\alpha}_{\overline{1}, \cdots, \overline{i{+}j_{n}{-}L}}|\vec{S}\rangle
    \label{SMeq: action_L_part_coupling_Lambda}\\
    &\hat{H}^{\alpha}_{\overline{i{+}j_{1}},\cdots,\overline{i{+}j_{n}}}|\Lambda\rangle_{\mathcal{C}}{=}\sum_{|\vec{S}\rangle}\mathcal{B}^{\alpha}_{1, \cdots, i{+}j_{n}{-}L}\mathcal{C}|\vec{S}\rangle{\otimes} \mathcal{A}^{\alpha}_{\overline{i{+}j_{1}}, \cdots, \overline{L}} |\vec{S}\rangle{=}{\mp}\sum_{|\vec{S}\rangle}\mathcal{C}(\mathcal{B}^{\alpha}_{1, \cdots, i{+}j_{n}{-}L})^{\ast}|\vec{S}\rangle{\otimes} \mathcal{A}^{\alpha}_{\overline{i{+}j_{1}}, \cdots, \overline{L}} |\vec{S}\rangle.
    \label{SMeq: action_R_part_coupling_Lambda}
\end{align}
\end{subequations}
Inserting in Eqs.~\eqref{SMeq: action_L_part_coupling_Lambda} and \eqref{SMeq: action_R_part_coupling_Lambda}, the resolution of identity $\mathbb{I}_{1,{\cdots}, L}{\otimes}\mathbb{I}_{\overline{1},{\cdots}, \overline{L}}{\equiv}\sum_{|\vec{F}\rangle,|\vec{G}\rangle \in\mathcal{H}_{L}}(|\vec{F}\rangle\langle \vec{F}|)_{1,{\cdots}, L}{\otimes} (|\vec{G}\rangle\langle \vec{G}|)_{\overline{1},{\cdots}, \overline{L}}$, we get
\begin{equation}
    \begin{split}
        \Big( H^{\alpha}_{i{+}j_{1},\cdots,i{+}j_{n}}{+}\hat{H}^{\alpha}_{\overline{i{+}j_{1}},\cdots,\overline{i{+}j_{n}}}\Big) |\Lambda\rangle_\mathcal{C}= {\pm}&
        \sum_{|\vec{S}\rangle,|\vec{F}\rangle,|\vec{G}\rangle} \Bigg[ \langle (\vec{F}|(\mathcal{A}^{\alpha}_{i{+}j_{1}, \cdots, L})^{\ast}|\vec{S}\rangle)_{1,{\cdots}, L} (\langle \vec{G}|\mathcal{B}^{\alpha}_{\overline{1}, \cdots, \overline{i{+}j_{n}{-}L}}|\vec{S}\rangle)_{\overline{1},{\cdots}, \overline{L}} -\\ &(\langle \vec{F}|(\mathcal{B}^{\alpha}_{1, \cdots, i{+}j_{n}{-}L})^{\ast}|\vec{S}\rangle)_{1,{\cdots}, L} (\langle \vec{G}|\mathcal{A}^{\alpha}_{\overline{i{+}j_{1}}, \cdots, \overline{L}}|\vec{S}\rangle)_{\overline{1},{\cdots}, \overline{L}}\Bigg]\mathcal{C}|\vec{F}\rangle\otimes|\vec{G}\rangle.
    \end{split}
\end{equation}
Using $(\langle \vec{G}|\mathcal{A}^{\alpha}_{\overline{i{+}j_{1}}, \cdots, \overline{L}}|\vec{S}\rangle)_{\overline{1},\cdots,\overline{L}}{=}(\langle \vec{G}|\mathcal{A}^{\alpha}_{{i{+}j_{1}}, \cdots, {L}}|\vec{S}\rangle)_{1,\cdots,L}$ (similarly for $\mathcal{B}^{\alpha})$ we obtain
\begin{equation}
    \begin{split}
        \Big( H^{\alpha}_{i{+}j_{1},\cdots,i{+}j_{n}}{+}\hat{H}^{\alpha}_{\overline{i{+}j_{1}},\cdots,\overline{i{+}j_{n}}}\Big) |\Lambda\rangle_\mathcal{C}= {\pm}
        \sum_{|\vec{S}\rangle,|\vec{F}\rangle,|\vec{G}\rangle} \Bigg[ &\langle \vec{F}|(\mathcal{A}^{\alpha}_{i{+}j_{1}, \cdots, L})^{\ast}|\vec{S}\rangle\langle \vec{G}|\mathcal{B}^{\alpha}_{{1}, \cdots, {i{+}j_{n}{-}L}}|\vec{S}\rangle -\\ &\langle \vec{F}|(\mathcal{B}^{\alpha}_{1, \cdots, i{+}j_{n}{-}L})^{\ast}|\vec{S}\rangle \langle \vec{G}|\mathcal{A}^{\alpha}_{{i{+}j_{1}}, \cdots, {L}}|\vec{S}\rangle\Bigg]\mathcal{C}|\vec{F}\rangle\otimes|\vec{G}\rangle.
    \end{split}
    \label{SMeq: Proof1_part1}
\end{equation}
Now the hermiticity of $\mathcal{A}^{\alpha}$ and 
$\mathcal{B}^{\alpha}$ implies $\langle \vec{F}|(\mathcal{A}^{\alpha}_{i{+}j_{1}, \cdots, L})^{\ast}|\vec{S}\rangle{=}\langle \vec{S}|\mathcal{A}^{\alpha}_{i{+}j_{1}, \cdots, L}|\vec{F}\rangle$ and $\langle \vec{F}|(\mathcal{B}^{\alpha}_{1, \cdots, i{+}j_{n}{-}L})^{\ast}|\vec{S}\rangle{=}$ $\langle \vec{S}|\mathcal{B}^{\alpha}_{1, \cdots, i{+}j_{n}{-}L}|\vec{F}\rangle$ respectively, leading us to
\begin{equation}
    \begin{split}
        \Big(H^{\alpha}_{i{+}j_{1},\cdots,i{+}j_{n}}{+}\hat{H}^{\alpha}_{\overline{i{+}j_{1}},\cdots,\overline{i{+}j_{n}}}\Big) |\Lambda\rangle_{\mC} = {\pm}
        \sum_{|\vec{F}\rangle,|\vec{S}\rangle,|\vec{G}\rangle} \Bigg[& \langle \vec{G}|\mathcal{B}^{\alpha}_{{1}, \cdots, {i{+}j_{n}{-}L}}|\vec{S}\rangle \langle \vec{S}|\mathcal{A}^{\alpha}_{i{+}j_{1}, \cdots, L}|\vec{F}\rangle-\\ & \langle \vec{G}|\mathcal{A}^{\alpha}_{{i{+}j_{1}}, \cdots, {L}}|\vec{S}\rangle\langle \vec{S}|\mathcal{B}^{\alpha}_{1, \cdots, i{+}j_{n}{-}L}|\vec{F}\rangle\Bigg]\mathcal{C}|\vec{F}\rangle\otimes|\vec{G}\rangle.
    \end{split}
    \label{SMeq: Proof1_part2}
\end{equation}
Since $\sum_{|\vec{S}\rangle \in \mathcal{H}_{L}}|\vec{S}\rangle\langle \vec{S}|{=}\mathbb{I}$, Eq.~\eqref{SMeq: Proof1_part2} implies $ \Big(\hat{H}^{\alpha}_{i+j_{1},\cdots,i+j_{n}}{+}\hat{H}^{\alpha}_{\overline{i{+}j_{1}},\cdots,\overline{i{+}j_{n}}}\Big) |\Lambda\rangle_{\mathcal{C}}{=}0$, completing the proof.
\end{proof}

\begin{lemma}
    For a spin-$S$ system, if $\mC_{l}$ is an invertible operator at site $l$ and $K_{l}$ is a non-zero complex number satisfying the condition
    \begin{equation}
        \mC_{l}[(S_{l}^{\mu})^{q}]^{\ast}=K_{l}(S_{l}^{\mu})^{q}\mC_{l},
        \label{SMeq: condition_K}
    \end{equation}
then $K_{l}$ can take only values $\pm 1$.
\label{Lemma1}
\end{lemma}

\begin{proof}
To simplify the analysis, we express each operator in Eq.~\eqref{SMeq: condition_K} as $(2S{+}1){\times}(2S{+}1)$ matrix in the basis where $S_{i}^{\mu}$ is diagonal. In this diagonal basis of $S_{i}^{\mu}$, comparing row-$m$ and column-$n$ element of both sides of Eq.~\eqref{SMeq: condition_K} gives the equation
\begin{equation}
    (\mathfrak{s}_{n})^{q}(\mathcal{C}_{l})_{mn}{=}K_{l}(\mathfrak{s}_{m})^{q}(\mathcal{C}_{l})_{mn} \hspace{0.5cm}\text{for all $m,n\in\{0,\cdots,2S\}
    $},
    \label{SMeq: comparing_each_entry}
\end{equation}
where $\mathfrak{s}_{m}{=}{-}S{+}m$ is the eigenvalue of $S_{i}^{\mu}$. Now to determine the possible values of $K_{l}$, we begin by dividing the potential values of $K_{l}$ (which are all non-zero complex numbers) into two distinct sets and analyzing the solutions separately: 
\begin{enumerate}
    \item[I.] The set containing all the non-zero finite $q$th power of the ratio of eigenvalues of $S_{i}^{\mu}$ i.e.
    \begin{equation}
        K_{l}=\Big(\frac{\mathfrak{s}_{a}}{\mathfrak{s}_{b}}\Big)^q\hspace{0.2cm}\text{where}\hspace{0.2cm}a,b\in\{0,\cdots,S{-}1,S{+}1,\cdots,2S\}.
    \label{SMeq: all_Kvalues_without_C}
    \end{equation}
    (Note that neither $a$ nor $b$ can equal $S$ since we are only considering non-zero finite $K_{l}$ in this set.)
    
    \item[II.] The set of all other non-zero complex numbers. 
    
\end{enumerate}
If $K_{l}$ is not of the form given in Eq.~\eqref{SMeq: all_Kvalues_without_C}, for Eq.~\eqref{SMeq: comparing_each_entry} to hold for all values $m$ and $n$, we must have $(\mC_{l})_{mn}{\times}(\mathfrak{s}_{n})^{q}{=}0$ [implies only non-zero element in $\mathcal{C}_{l}$ can be $(\mathcal{C}_{l})_{m,S}$ since $\mathfrak{s}_{S}{=}0$] and $(\mC_{l})_{mn}{\times}(\mathfrak{s}_{m})^{q}{=}0$ [implies only non-zero element in $\mathcal{C}_{l}$ can be $(\mathcal{C}_{l})_{S,n}$ since $\mathfrak{s}_{S}{=}0$] for all $m,n$. This implies that the matrix elements of $\mC_{l}$ are of the form $(\mC_{l})_{mn}{=}C_{S}\delta_{m, S}\delta_{n, S}$ where $C_{S}$ is an arbitrary complex number [i.e., the only non-zero element in $\mathcal{C}_{l}$ can be ($\mathcal{C}_{l})_{S, S}$], and hence $\det[\mC_{l}]{=}0$ for any value of $C_{S}$. This means that $\mC_l$ must be non-invertible, violating the assumptions of the Lemma~\ref{Lemma1}. Thus the values of $K_{l}$ in set (II) are not valid solutions. 

Now from the set of values of $K_{l}$ given in Eq.~\eqref{SMeq: all_Kvalues_without_C}, we want to find the values of $K_{l}$ consistent with $\mC_l$ being invertible. The eigenvalues $\mathfrak{s}_{a}$ and $\mathfrak{s}_{b}$ satisfy one of the following relations
\begin{enumerate}
    \item[(i)] $\mathfrak{s}_{a}{=}\mathfrak{s}_{b}$ (when $a{=}b$). This would mean that $K_l{=}{+}1$.
    \item[(ii)] $\mathfrak{s}_{a}{=}{-}\mathfrak{s}_{b}$ (when $a{=}2S{-}b$). This would mean that $K_l{=}{-}1$ when $q$ is odd and $K_l = +1$ when $q$ is even.
    \item[(iii)] $|\mathfrak{s}_{a}|{\neq}|\mathfrak{s}_{b}|$ (otherwise). This would mean that $K_l{\neq}{\pm}1$.
\end{enumerate}
Therefore, we can divide the range of values of $K_{l}$ in Eq.~\eqref{SMeq: all_Kvalues_without_C} into three categories. We then solve for $\mC_{l}$ using Eq.~\eqref{SMeq: comparing_each_entry} for each $K_{l}$ in these categories and check if the invertibility condition is met. The detailed procedure is outlined below.
\begin{itemize}
    \item \textbf{Category (i):} 
    Using $K_{l}{=}1$ in Eq.~\eqref{SMeq: comparing_each_entry} we get,
\begin{equation}
    (\mathfrak{s}_{n})^{q}(\mathcal{C}_{l})_{mn}{=}(\mathfrak{s}_{m})^{q}(\mathcal{C}_{l})_{mn} \hspace{0.5cm}\text{for all $m,n\in\{0,\cdots,2S\}
    $}.
    \label{SMeq: comparing_each_entry_for_K=1}
\end{equation}
Now solving Eq.~\eqref{SMeq: comparing_each_entry_for_K=1} for all $m,n$, we find  the matrix elements of $\mC_{l}$ can be given by $(\mC_{l})_{mn}{=}C_{m}\delta_{mn}$, where $\{C_{m}\}$ are arbitrary complex numbers. For any non-zero set of $\{C_{m}\}$, the invertibility condition $\text{det}[\mC_{l}]{\neq}0$ can be satisfied, hence $K_{l}{=}1$ is a potentially valid solution for $K_{l}$. 
\item \textbf{Category (ii):}
Substituting $K_{l}{=}{\pm}1$ into Eq.~\eqref{SMeq: comparing_each_entry} for all values of $m,n$, we find the matrix elements $(\mC_{l})_{mn}{=}C_{n}\delta_{2S{-}m,n}$, where $\{C_{n}\}$ are arbitrary complex numbers. Similarly, any non-zero set of $\{C_{n}\}$ also satisfies $\text{det}[\mC_{l}]{\neq}0$ consistent with $\mathcal{C}$ being invertible, so $K_{l}{=}{\pm}1$ is also a potentially valid solution for $K_{l}$. 
    \item \textbf{Category (iii)}:
From Eq.~\eqref{SMeq: all_Kvalues_without_C} let us consider an arbitrary value of $K_{l}$ given by $K_{l}{=} (\mathfrak{s}_{a_{1}})^{q}/(\mathfrak{s}_{b_{1}})^{q}$ in this category (i.e., with $a_{1}{\neq}b_{1}$ and $a_{1}{\neq}2S{-}b_{1}$). Substituting this value of $K_{l}$ back into Eq.~\eqref{SMeq: comparing_each_entry}, we get the matrix element of $\mC_{l}$:
\begin{equation}
    (\mC_{l})_{mn}{=}C_{mn}(\delta_{m,a_{1}}\delta_{n,b_{1}}{+}\delta_{m,2S{-}a_{1}}\delta_{n,2S{-}b_{1}}),
    \label{SMeq: non-inverible C}
\end{equation}
where $\{C_{mn}\}$ are arbitrary complex numbers. Any set of values of $\{C_{mn}\}$ leads to $\text{det}[\mC_{l}]{=}0$ making it an invalid solution for $K_{l}$ under the assumptions of the Lemma~\ref{Lemma1}. 
\end{itemize}
Thus, with the invertibility constraint on $\mC_{l}$, the only valid non-zero values for $K_{l}$ are ${\pm}1$, completing the proof of the Lemma~\ref{Lemma1}.
\end{proof}

\section{Generalization of Theorem~\ref{SMtheorem: 1} to higher dimensions and applications to the nearest neighbor spin-$S$ $XY$ model on a hypercube}
\label{sec: generalization_and_new_examples}
In this section, we extend Theorem~\ref{SMtheorem: 1} to two-dimensional rectangular systems. Then, applying the generalization of Theorem~\ref{SMtheorem: 1}, we provide explicit examples of zero energy eigenstates for the two-dimensional spin-$S$ $XY$ model with a transverse field that we considered in the main text. We then show some examples to suggest how our results can be generalized to (hyper)cuboidal lattices in arbitrary dimensions. 

We consider the class of spin-$S$ systems on rectangular grids of size $M{\times}L$ with a translation invariant interaction [with PBC in both directions] of the form
\begin{equation}
    H^{2D}(M{\times}L){=}\sum_{\alpha,\{r_{i,j}\}} \hat{H}^{\alpha}_{\{r_{i,j}\},n}
    {\equiv}\sum_{\substack{\alpha,\\1{\leq}i{\leq}M,\\1{\leq}j{\leq}L}} \hat{H}^{\alpha}_{(r_{i,j})_{1},\cdots,(r_{i,j})_{n}}{=}\sum_{\substack{\alpha,\\1{\leq}i{\leq}M,\\1{\leq}j{\leq}L}}\prod_{k{=}1}^{n} (S_{(r_{i,j})_{k}}^{\mu_k})^{q_{k}},
    \label{SMeq: generic_Hamiltonian_2D}
\end{equation}
where $S_{(r_{i,j})_{k}}^{\mu_k}(\mu_k{\in} \{x,y,z\})$ are the spin operators at site $(i,j)$ indexed by $r_{i,j}{=}i{+}M{\times}(j{-}1)$ with $i{\in}\{1, {\cdots}, M\}$  and $j{\in}\{1, {\cdots}, L\}$ labeling the horizontal and vertical directions, respectively, with $(1,1)$ at the bottom-left corner and each $q_k$ is a non-negative integer (For ease of notation, we have used $\mu_{k}{\equiv}\mu_{(r_{i,j})_{k}}$ and $q_{k}{\equiv}q_{(r_{i,j})_{k}}$.). The notation $(r_{i,j})_n$ is used to denote the $n^{\rm th}$ site in the vicinity of the site $r_{i,j}$ with $(r_{i,j})_1{\equiv}r_{i,j}$. The term $\hat{H}^{\alpha}_{(r_{i,j})_{1},\cdots,(r_{i,j})_{n}}$ represents an $\alpha$-type (determined by the specific values of $\{\mu_{k}\}$ and $\{q_{k}\}$) interaction term between $n{-}$spins at the sites $\{~(r_{i,j})_{1},{\cdots},(r_{i,j})_{n}~\}$. 
\subsection{General Theorem}
For the higher-dimensional setting described above, the result generalizing Theorem~\ref{SMtheorem: 1} is as follows:
\begin{thm}
\label{generalization_Theorem_1_to_2D}
If each $\hat{H}_{\{r\}}^{\alpha}$ term in the Hamiltonian $H^{2D}(M{\times}L)$ anti-commutes with an operator $\mathcal{C}\mathcal{K}$ [analogous to the condition stated in Eq.~\eqref{SMeq: anti_commuation}], where $\mathcal{C}{=}\prod_{i,j}\mathcal{C}_{r_{i,j}}$ with $\mathcal{C}_{r_{i,j}}$ on-site invertible operators, and $\mathcal{K}$ is the complex conjugation operator, then the state
\begin{equation}
\label{SMeq: zero_energystate_2d}
\begin{split}
    |\Lambda\rangle_{\mathcal{C}}{=}{\sum\limits_{|\vec{S}\rangle\in\mathcal{H}_{M{\times}L}}}\mathcal{C}|\vec{S}\rangle_{\{r_{i,j}\}}  {\otimes} |\vec{S}\rangle_{\{\overline{r_{i,j}}\}}
    {=}{\bigotimes_{i,j}^{M,L}}~{\sum_{s{=}1}^{2S{+}1}}\mathcal{C}_{r_{i,j}}|s_{r_{i,j}}{\rangle}|s_{\overline{r_{i,j}}}{\rangle},
\end{split}
\end{equation}
where $\overline{r_{i,j}}{\equiv}r_{i,j{+}L}$ [$\overline{r_{i,j}}{\equiv}r_{i{+}L,j}$], is an exact zero-energy eigenstate of the Hamiltonian $H(M{\times}2L)$ [resp. $H(2M{\times}L)$] on a system of size $M{\times}2L$ [resp. $2M{\times}L$] with PBC in both directions.
\end{thm}
The proof of Theorem~\ref{generalization_Theorem_1_to_2D} proceeds along the same lines as that of Theorem~\ref{SMtheorem: 1}.
Using the mapping $\overline{r_{i,j}}{\equiv}r_{i,j{+}L}$, the Hamiltonian for the $M{\times}2L$ system with PBC can be expressed as $H(M{\times}2L){=}\sum_{\alpha,i,j}^{M,L} (\hat{H}^{\alpha}_{\{r_{i,j}\}, n}{+}\hat{H}^{\alpha}_{\{\overline{r_{i,j}}\}, n})$, and then for the state $|\Lambda\rangle_{\mathcal{C}}$, the action of the term $\hat{H}^{\alpha}_{\{r_{i,j}\}, n}$ is canceled by that of $\hat{H}^{\alpha}_{\{\overline{r_{i,j}}\}, n}$ for all $i,j$.
Similar ideas can be used to establish the analogous result for the $2M{\times}L$ system.

\begin{figure}
    \centering
    \includegraphics[scale=0.4]{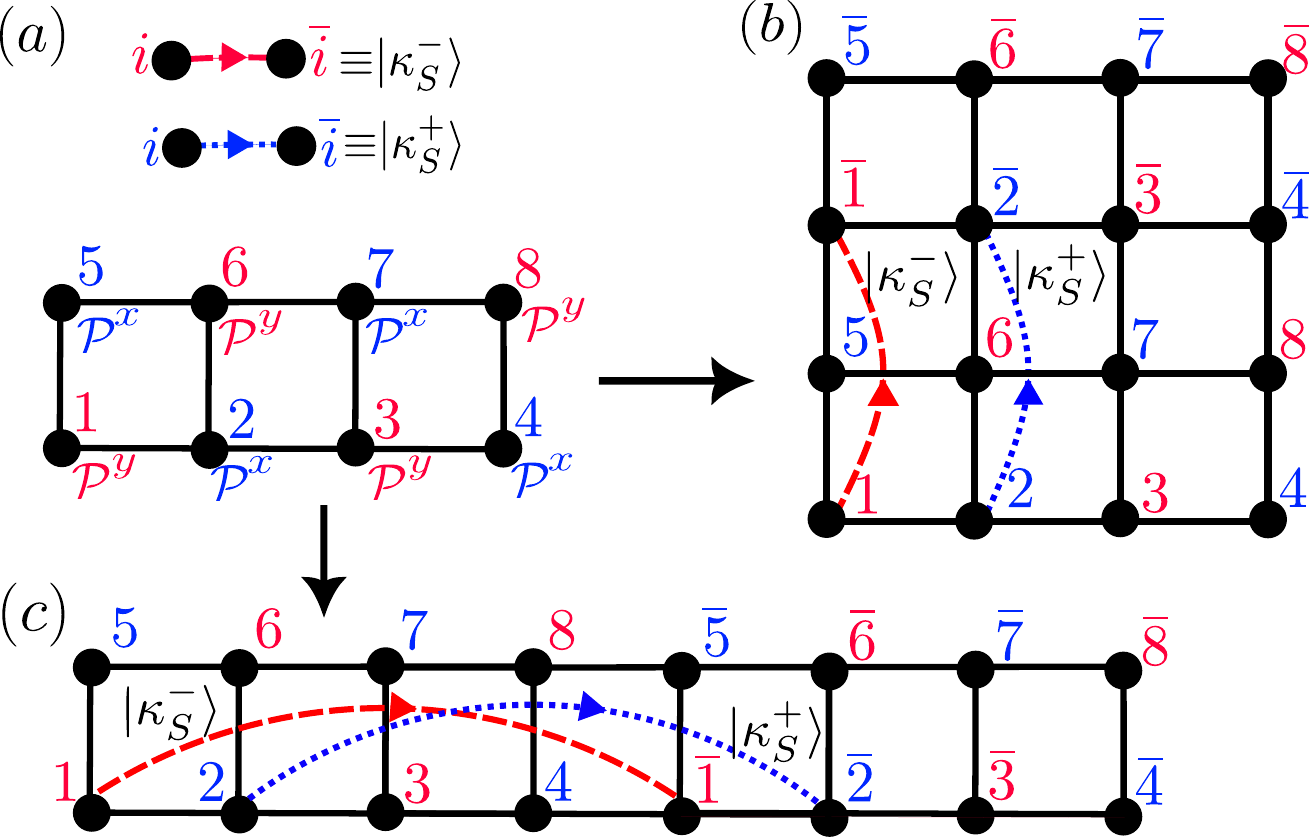}
    \caption{Schematic representation of the zero energy eigenstate $|\Lambda\rangle_\mathcal{C}$ [defined in Eq.~\eqref{SMeq: zero_energystate_2d}] for nearest neighbor spin-$S$ $XY$ model with a transverse field in two dimensions obtained using Theorem~\ref{generalization_Theorem_1_to_2D}. The black dots denote the sites and the solid lines between them indicate the $XY$ interaction. The spins at red (blue) sites $r$ are entangled only with their partners $\overline{r}$ in the state $|\kappa_{S}^{-}\rangle$ ($|\kappa_{S}^{+}\rangle$) [defined in Eq.~\eqref{SMeq: kappa}]. $(a)$ shows $\mathcal{C}_{XY}^{2D}$ satisfies Eq.~\eqref{SMeq: condition} for a $M{\times}L{=}4{\times}2$ rectangular lattice with PBC, which gives a zero energy state for $(b)$ $M{\times}2L$ lattice and $(c)$ $2M{\times}L$ lattice with PBC.}
    \label{fig: supp_fig1}
\end{figure}

\begin{figure}
    \centering
    \includegraphics[scale=0.4]{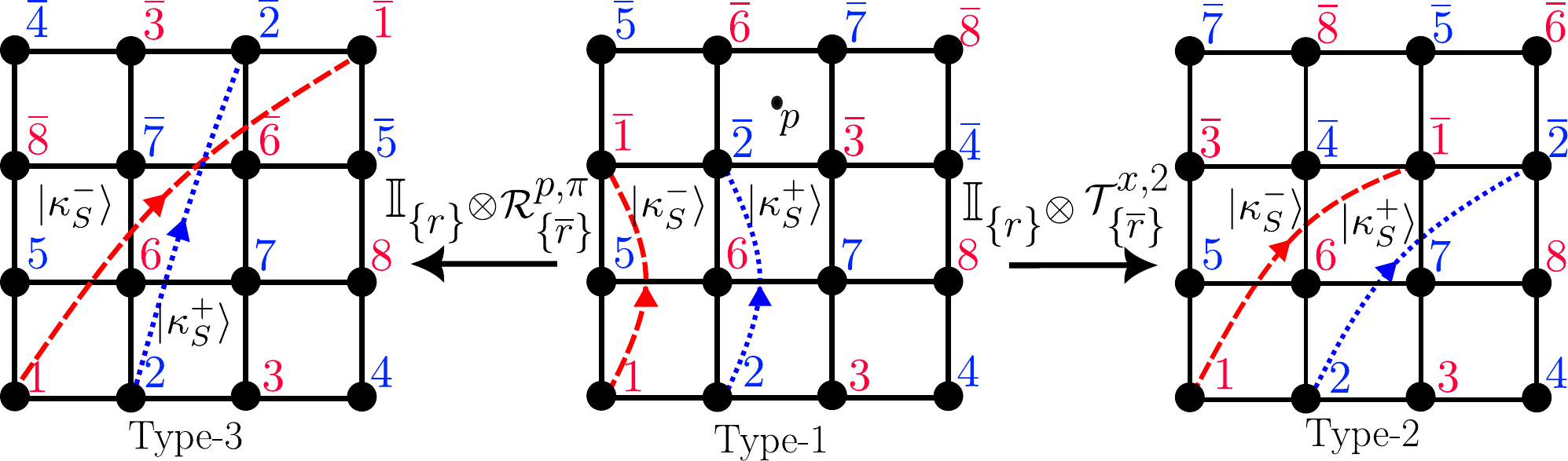}
    \caption{Schematic representation of the different types of zero energy eigenstate of the form $|\Lambda\rangle_\mathcal{C}$ [defined in Eq.~\eqref{SMeq: zero_energystate_2d}] of nearest neighbor spin-$S$ $XY$ model with transverse field defined on a $M{\times}2L{=}4{\times}4$ square lattice with PBC. The black dots denote the sites and the solid lines between them indicate the $XY$ interaction. The spins at red (blue) sites $r$ are entangled only with their partners $\overline{r}$ in the state $|\kappa_{S}^{-}\rangle$ ($|\kappa_{S}^{+}\rangle$) [defined in Eq.~\eqref{SMeq: kappa}]. We get the type-1 state as a direct consequence of Theorem~\ref{generalization_Theorem_1_to_2D}, while the type-2 and type-3 states are constructed by acting $\mathbb{I}_{\{r\}}{\otimes}\mathcal{T}^{x,M/2}_{\{\overline{r}\}}$ and $\mathbb{I}_{\{r\}}{\otimes}\mathcal{R}^{p,\pi}_{\{\overline{r}\}}$ on the type-1 state. Here the operator $\mathcal{T}^{x,M/2}_{\{\overline{r}\}}$ translates all the $\{\overline{r}\}$ sites by $M/2$ units in the horizontal direction and the operator $\mathcal{R}^{p,\pi}_{\{\overline{r}\}}$ rotates all the $\{\overline{r}\}$ sites around the point $p{=}[(M{+}1)/2,(3L{+}1)/2]$ by an angle $\pi$.}
    \label{fig: supp_fig2}
\end{figure}

\subsection{Concrete Example}
Next, we provide concrete examples of states obtained using Theorem~\ref{generalization_Theorem_1_to_2D} in the two-dimensional spin-$S$ $XY$ model. Consider the spin-$S$ $XY$ model with a transverse field on an $(M{\times}L)$ rectangular lattice with PBC in both directions described by the Hamiltonian
\begin{equation}
\begin{split}
    H_{XY}^{2D}(M{\times}L){=}\sum_{\eta,i,j}J_{\eta}(S_{r_{i,j}}^{\eta}S_{r_{i{+}1,j}}^{\eta}{+}S_{r_{i,j}}^{\eta}S_{r_{i,j{+}1}}^{\eta}){+}h\sum_{i,j}S_{r_{i,j}}^{z},    
\end{split}
\label{SMeq: 2D XY model Hamiltonian}
\end{equation}
where $\eta{\in}\{x,y\}$ represents the $XY$-interaction and $J_{\mu}$ and $h$ are arbitrary couplings. 
Using the mapping $\overline{r_{i,j}}{=}r_{i,j{+}L}$, the $XY$ model Hamiltonian for a $M{\times}2L$ lattice with PBC can be expressed as
\begin{equation}
\begin{split}
    H&_{XY}^{2D}(M{\times}2L){=}\sum_{\eta,i,j}\Big[J_{\eta}(S_{r_{i,j}}^{\eta}S_{r_{i{+}1,j}}^{\eta}{+}S_{\overline{r_{i,j}}}^{\eta}S_{\overline{r_{i{+}1,j}}}^{\eta})\\
    &{+}J_{\eta}(S_{r_{i,j}}^{\eta}S_{r_{i,j{+}1}}^{\eta}{+}S_{\overline{r_{i,j}}}^{\eta}S_{\overline{r_{i,j{+}1}}}^{\eta}){+}h(S_{r_{i,j}}^{z}{+}S_{\overline{r_{i,j}}}^{z})\Big].    
\end{split}
\label{SMeq: 2D XY model Hamiltonian_relabelled}
\end{equation}
In the computational basis $|m\rangle_{r_{i,j}}$ that are eigenstates of $S_{r_{i,j}}^{z}$, each $\hat{H}_{\{r\}}^{\alpha}$ of the Hamiltonian, i.e., $S_{r_{i,j}}^{x}S_{r_{i{+}1,j}}^{x},S_{r_{i,j}}^{y}S_{r_{i{+}1,j}}^{y}, $\\$S_{r_{i,j}}^{x}S_{r_{i,j{+}1}}^{x}$, $S_{r_{i,j}}^{y}S_{r_{i,j{+}1}}^{y}$ and $S_{r_{i,j}}^{z}$ is real.
For even values of $M$ and $L$, they also anticommute with the operator [shown in Fig.~\ref{fig: supp_fig1}$(a)$]
\begin{equation}
    \mathcal{C}_{XY}^{2D}{=}\prod_{i,j}(\mathcal{P}_{r_{i,j}}^{y})^{\delta_{(i{+}j){\rm~mod~}2,0}}(\mathcal{P}_{r_{i,j}}^{x})^{\delta_{(i{+}j){\rm~mod~}2,1}}.
    \label{SMeq: operator_2DXY}
\end{equation}
Thus, using Theorem~\ref{generalization_Theorem_1_to_2D}, the Hamiltonian $H_{XY}^{2D}(M{\times}2L)$ hosts a zero energy eigenstate of the form $|\Lambda\rangle_{C}$ which can be expressed as 
\begin{eqnarray}
    |\Lambda_{1}\rangle_{\mathcal{C}_{XY}}^{2D}{=}\bigotimes_{i,j}(|\kappa_{S}^{-}\rangle_{r,\overline{r}})^{\delta_{(i{+}j){\rm~mod~}2,0}}
    (|\kappa_{S}^{+}\rangle_{r,\overline{r}})^{\delta_{(i{+}j){\rm~mod~}2,1}},
    \label{SMeq: zeroenergy_eigstate_2DXY}\\
    |\kappa_{S}^{\pm}\rangle_{i,j}{=}\frac{1}{\sqrt{2S{+}1}}\sum_{p{=}0}^{2S}({\pm}1)^{p}(S^{-})^{p}|S\rangle_{i}{\otimes}(S^{+})^{p}|{-}S\rangle_{j}.
   \label{SMeq: kappa}
\end{eqnarray}
The state $|\Lambda\rangle_{\mathcal{C}_{XY}}^{2D}$ is shown in Fig.~\ref{fig: supp_fig1}$(b)$, where the spins at sites $r_{i,j}$ and $\overline{r_{i,j}}$ are in the state $|\kappa_{S}^{-}\rangle$ when $i{+}j$ is even (marked by red dashes) and in the state $|\kappa_{S}^{+}\rangle$ when $i{+}j$ is odd (marked by blue dots). Similarly, Fig.~\ref{fig: supp_fig1}$(c)$ shows a zero energy eigenstate for $H_{XY}^{2D}(2M{\times}L)$. 
\subsection{Different classes of states}
Furthermore, using lattice symmetry operations, there are additional classes of zero-energy eigenstates for the two-dimensional spin-$S$ $XY$ model that can be constructed.
We now present some examples for the $M{\times}2L$ rectangular lattice.
Consider the state 
\begin{equation}
  |\Lambda_{2}\rangle_{\mathcal{C}_{XY}}^{2D}{=}\mathbb{I}_{\{r\}}{\otimes}\mathcal{T}^{x,M/2}_{\{\overline{r}\}}|\Lambda_{1}\rangle_{\mathcal{C}_{XY}}^{2D}{=}\bigotimes_{i,j}(|\kappa_{S}^{-}\rangle_{r,\overline{r}})^{\delta_{(i{+}j){\rm~mod~}2,0}}
    (|\kappa_{S}^{+}\rangle_{r,\overline{r}})^{\delta_{(i{+}j){\rm~mod~}2,1}}~~\text{with}~~\overline{r_{i,j}}{\equiv}r_{i{+}M/2,j{+}L}
  \label{SMeq: 2D: type2}
\end{equation}
(shown in the right part of Fig.~\ref{fig: supp_fig2}), where the operator $\mathcal{T}^{x,M/2}_{\{\overline{r}\}}$ translates all the $\{\overline{r}\}$ sites by $M/2$ units [recall $M$ is even] in the horizontal direction.
With the mapping $\overline{r_{i,j}}{\equiv}r_{i{+}M/2,j{+}L}$, the Hamiltonian can be cast in the form given in Eq.~\eqref{SMeq: 2D XY model Hamiltonian_relabelled} while the operator $\mathcal{C}_{XY}^{2D}$ given in Eq.~\eqref{SMeq: operator_2DXY} still satisfies the required condition Eq.~\eqref{SMeq: condition} [for example, as shown in right part of Fig.~\ref{fig: supp_fig2}, the Hamiltonian for $(M{\times}2L){=}(4 {\times} 4)$ contains terms of the form $(S_{1}^{\mu}S_{2}^{\mu}{+}S_{\overline{1}}^{\mu}S_{\overline{2}}^{\mu}),(S_{1}^{\mu}S_{5}^{\mu}{+}S_{\overline{1}}^{\mu}S_{\overline{5}}^{\mu}),(S_{3}^{\mu}S_{\overline{5}}^{\mu}{+}S_{\overline{3}}^{\mu}S_{5}^{\mu})$ and all of them anti-commute with $\mathcal{C}_{XY}^{2D}$ given in Eq.~\eqref{SMeq: operator_2DXY}].
Consequently $|\Lambda_{2}\rangle_{\mathcal{C}_{XY}}^{2D}$ is a zero energy eigenstate.
Similarly, we can demonstrate that the state
\begin{equation}
   |\Lambda_{3}\rangle_{\mathcal{C}_{XY}}^{2D}{=}\mathbb{I}_{\{r\}}{\otimes}\mathcal{R}^{p,\pi}_{\{\overline{r}\}}|\Lambda_{1}\rangle_{\mathcal{C}_{XY}}^{2D}{=}\bigotimes_{i,j}(|\kappa_{S}^{-}\rangle_{r,\overline{r}})^{\delta_{(i{+}j){\rm~mod~}2,0}}
    (|\kappa_{S}^{+}\rangle_{r,\overline{r}})^{\delta_{(i{+}j){\rm~mod~}2,1}}~~\text{with}~~\overline{r_{i,j}}{\equiv}r_{M{+}1{-}i,2L{+}1{-}j}
   \label{SMeq: 2D type3}
\end{equation}
(shown in the left part of Fig.~\ref{fig: supp_fig2}) is also a zero energy eigenstate of $H_{XY}^{2D}$, where the operator $\mathcal{R}^{p,\pi}_{\{\overline{r}\}}$ rotates all the $\{\overline{r}\}$ sites around the point $p{=}[(M{+}1)/2,(3L{+}1)/2]$ by an angle $\pi$.
We have labelled the states $|\Lambda_{1}\rangle_{\mathcal{C}_{XY}}^{2D}$, $|\Lambda_{2}\rangle_{\mathcal{C}_{XY}}^{2D}$ and $|\Lambda_{3}\rangle_{\mathcal{C}_{XY}}^{2D}$ as type-1, type-2, and type-3 states respectively in Fig~\ref{fig: supp_fig2}, all of which originate from the same configuration shown in Fig.~\ref{fig: supp_fig1}(a).

\begin{figure}
    \centering
    \includegraphics[scale=0.4]{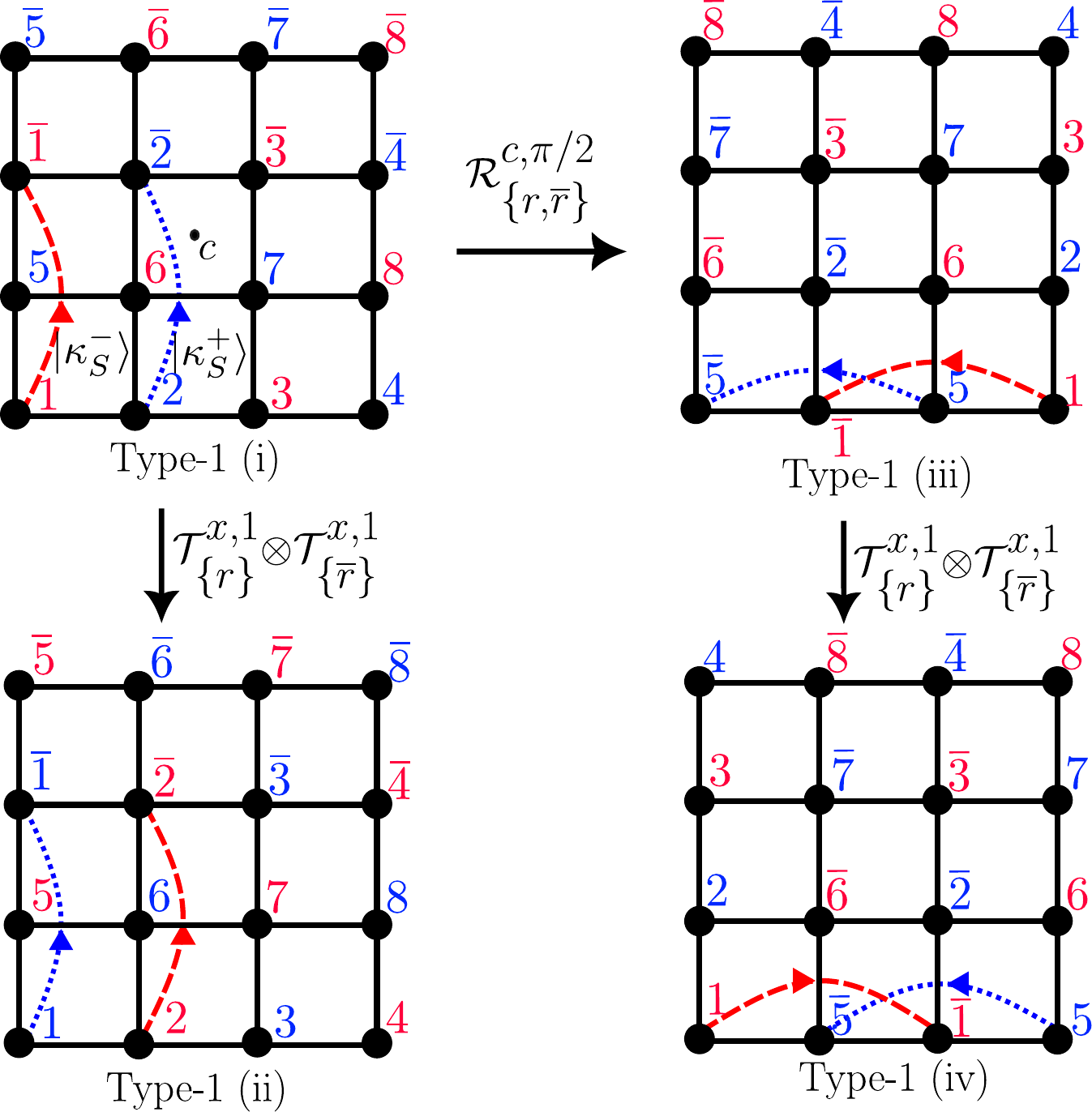}
    \caption{Schematic representation of the different kinds of type-1 zero energy eigenstate for nearest neighbor spin-$S$ $XY$ model with a transverse field on a $M{\times}2L{=}4{\times}4$ square lattice with PBC obtained by applying various symmetry operations on the original state: translation, rotation ($\mathcal{R}^{c,\pi/2}_{\{r,\overline{r}\}}$: rotating all the sites about the center $c{=}[(M{-}1)/2,(2L{-}1)/2]$ of the lattice by an angle $\pi/2$). The black dots denote the sites and the solid lines between them indicate the $XY$ interaction. The spins at red (blue) sites $r$ are entangled only with their partners $\overline{r}$ in the state $|\kappa_{S}^{-}\rangle$ ($|\kappa_{S}^{+}\rangle$) [defined in Eq.~\eqref{SMeq: kappa}].}
    \label{fig: supp_fig3}
\end{figure}

\begin{figure}[h!]
    \centering
    \includegraphics[scale=0.4]{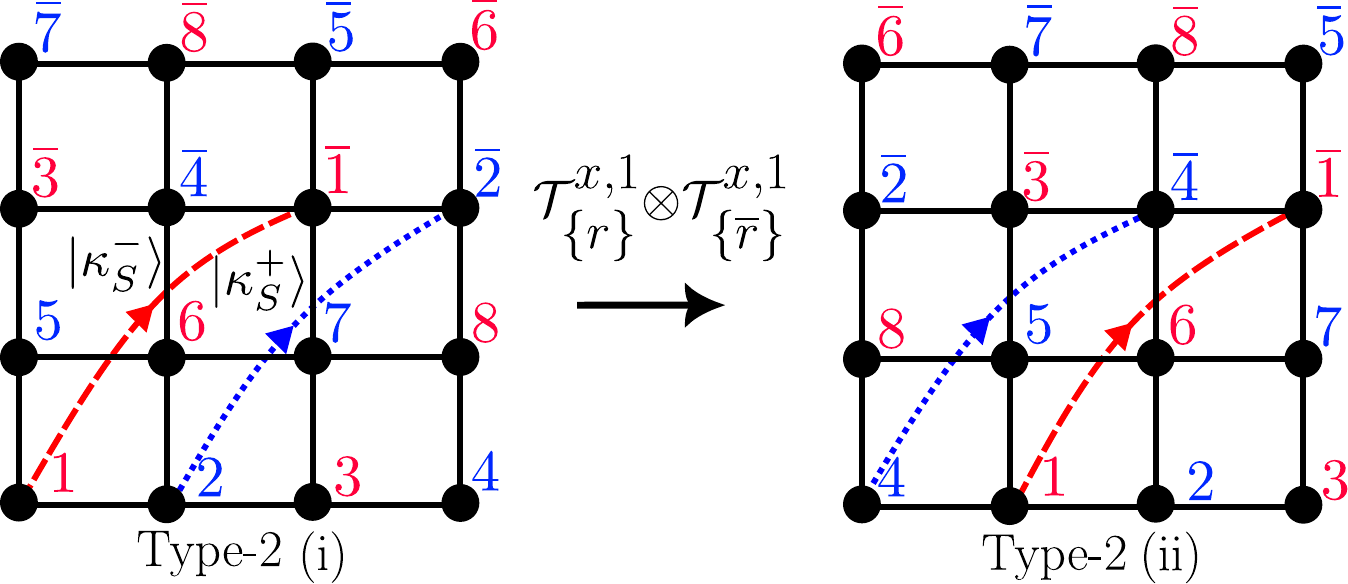}
    \caption{Schematic representation of the different kinds of type-2 zero energy eigenstates for nearest neighbor spin-$S$ $XY$ model with a transverse field on a $M{\times}2L{=}4{\times}4$ square lattice with PBC. The black dots denote the sites and the solid lines between them indicate the $XY$ interaction. The spins at red (blue) sites $r$ are entangled only with their partners $\overline{r}$ in the state $|\kappa_{S}^{-}\rangle$ ($|\kappa_{S}^{+}\rangle$) [defined in Eq.~\eqref{SMeq: kappa}].}
    \label{fig: supp_fig4}
\end{figure}

\begin{figure}[h!]
    \centering
    \includegraphics[scale=0.4]{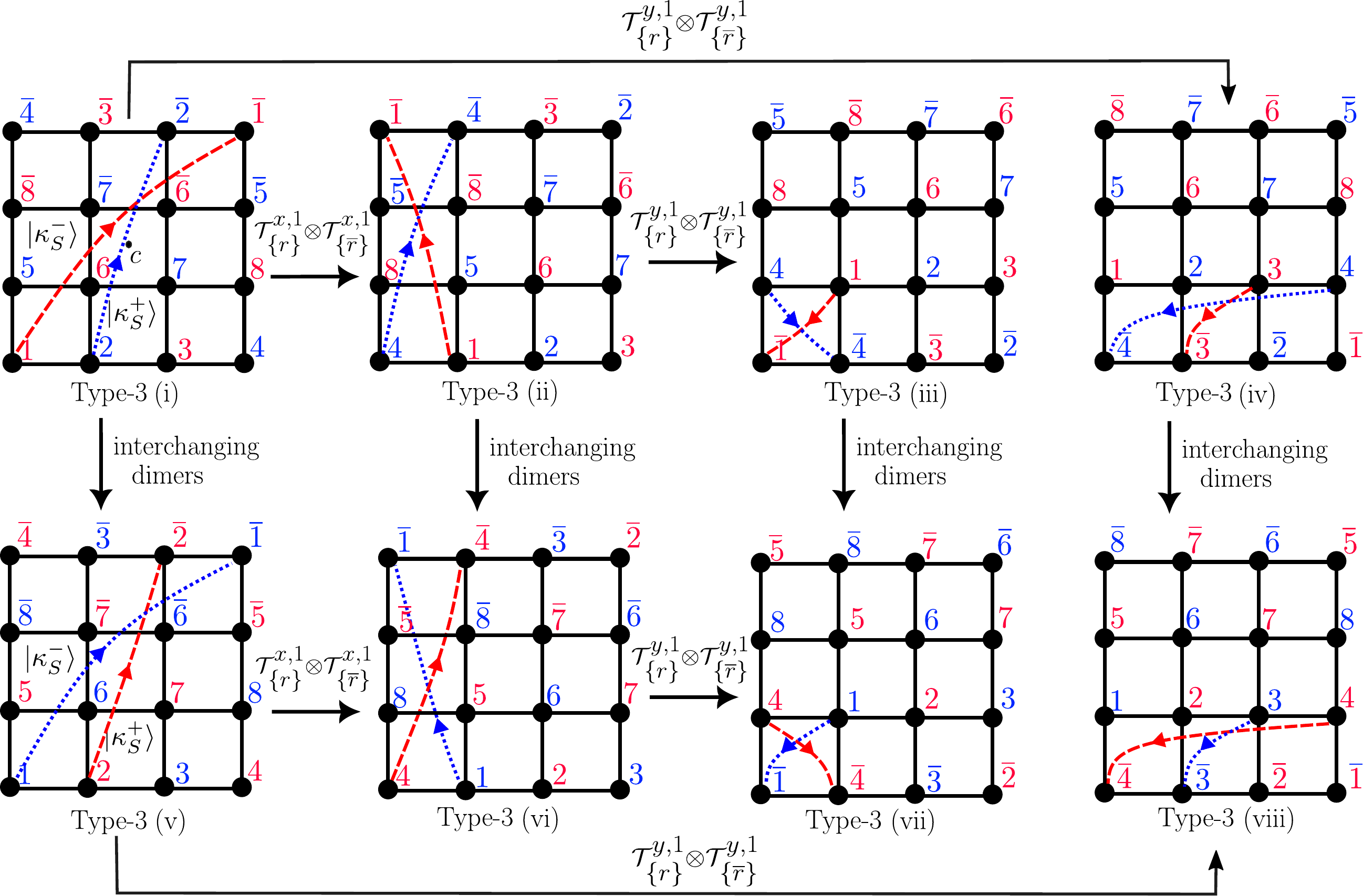}
    \caption{Schematic representation of the different kinds of type-3 of zero energy eigenstate for nearest neighbor spin-$S$ $XY$ model with a transverse field on a $M{\times}2L{=}4{\times}4$ square lattice with PBC. The black dots denote the sites and the solid lines between them indicate the $XY$ interaction. The spins at red (blue) sites $r$ are entangled only with their partners $\overline{r}$ in the state $|\kappa_{S}^{-}\rangle$ ($|\kappa_{S}^{+}\rangle$) [defined in Eq.~\eqref{SMeq: kappa}].}
    \label{fig: supp_fig5}
\end{figure}

In addition, various global symmetry operations, such as translations, rotations, and interchanging the dimers, generate distinct numbers of states for each type. The precise number of states for each type depends on both the symmetries and the structure of the states themselves. Figs.~\ref{fig: supp_fig3}, \ref{fig: supp_fig4}, and \ref{fig: supp_fig5} illustrate the number of states for different types. In general, for an $(M{\times}2L)$ lattice with PBC, there are 
\begin{itemize}
    \item two distinct type-1 states: $|\Lambda_{1}\rangle_{\mathcal{C}_{XY}}^{2D}$ and $(\mathcal{T}^{x,1}_{\{r\}}{\otimes}\mathcal{T}^{x,1}_{\{\overline{r}\}})|\Lambda_{1}\rangle_{\mathcal{C}_{XY}}^{2D}$ \\(additional two states: $\mathcal{R}^{c,\pi/2}_{\{r,\overline{r}\}}|\Lambda_{1}\rangle_{\mathcal{C}_{XY}}^{2D}$ and $(\mathcal{T}^{x,1}_{\{r\}}{\otimes}\mathcal{T}^{x,1}_{\{\overline{r}\}})\mathcal{R}^{c,\pi/2}_{\{r,\overline{r}\}}|\Lambda_{1}\rangle_{\mathcal{C}_{XY}}^{2D}$ when $M{=}2L$, where $\mathcal{R}^{c,\pi/2}_{\{r,\overline{r}\}}$ represents rotating all the sites about the center $c{=}[(2L{-}1)/2,(2L{-}1)/2]$ of the lattice by an angle $\pi/2$),
    \item two distinct type-2 states: $|\Lambda_{2}\rangle_{\mathcal{C}_{XY}}^{2D}$ and $(\mathcal{T}^{x,1}_{\{r\}}{\otimes}\mathcal{T}^{x,1}_{\{\overline{r}\}})|\Lambda_{2}\rangle_{\mathcal{C}_{XY}}^{2D}$,
    \item $ML/2$ distinct type-3 states: $(\mathcal{T}^{x,m}_{\{r\}}{\otimes}\mathcal{T}^{x,m}_{\{\overline{r}\}})(\mathcal{T}^{y,l}_{\{r\}}{\otimes}\mathcal{T}^{y,l}_{\{\overline{r}\}})|\Lambda_{3}\rangle_{\mathcal{C}_{XY}}^{2D}$ with $m{=}\{0,1,{\cdots},M/2{-}1\}$ and $n{=}\{0,1,{\cdots},L{-}1\}$.\\
    Furthermore, interchanging the dimers in each of the above states can result in an additional $ML/2$ states.
\end{itemize}
Thus, it is clear that the total number of zero-energy states scales with the volume of the system $\mathcal{O}(ML)$. Note that even though we have only demonstrated the existence of a large number of type-3 states for the spin-$S$ $XY$ model explicitly, these states also exist for any nearest-neighbor 2D Hamiltonian of the form defined in Eq.~\eqref{SMeq: generic_Hamiltonian_2D} and satisfying the condition given in Eq.~\eqref{SMeq: condition}.

\begin{comment}
\sanjay{Do we know that these are the only types of states (I think we discussed this, but I forgot the answer)? If we don't have a clear proof, we can propose a simple ``algorithm" for finding all possible states for any nearest-neighbor Hamiltonian on any graph $G$:
\begin{enumerate}
%
\item Enumerate all pairings of vertices (i.e., half of the vertices are mapped on to the other half of the vertices) -- I think there should in general be $(|V|-1)!!$ distinct pairings. 
%
\item A pairing between vertices also gives rise to a mapping between edges of the graph. Any pairing is said to be ``admissible" if all edges of the graph are mapped into other edges of the graph (this is to ensure that there is a different term in the Hamiltonian that cancels the action of the first term). 
%
\item I think one should be able to construct a state for any admissible pairing.
%
\end{enumerate}
%
Unless you have already implemented something like this, I wonder if it is feasible to do a quick check for a small graph (e.g., 4 x 4 lattice) to see if we get something different from these three distinct types.
%
There might be some naive speedups one can get from the symmetries of the lattice.
}
\end{comment}

Our ideas can be readily generalized to three and higher dimensions. Certain zero-energy states of the type $|\Lambda\rangle_{\mathcal{C}}$ for the spin-$S$ $XY$ model with a transverse field on a $(M{\times}2L{\times}N)$ cuboid is shown in Fig.~\ref{fig: supp_fig6}.
A systematic investigation of the various classes, in particular, precisely enumerating the number of distinct classes and the corresponding number of states within each class, is an interesting direction that we leave for future exploration.
\begin{figure}
    \centering
    \includegraphics[scale=0.4]{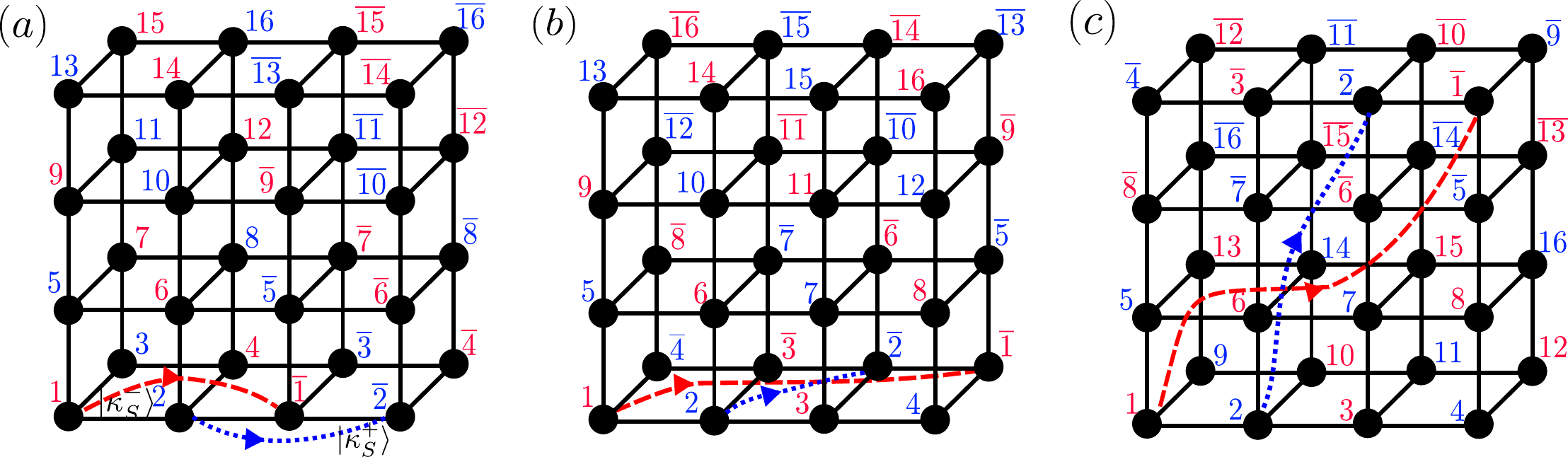}
    \caption{Schematic representation of some types of zero energy eigenstate of nearest neighbor spin-$S$ $XY$ model with a transverse field defined on a $M{\times}2L{\times}N{=}4{\times}4{\times}2$ cuboid with PBC in all directions. The black dots denote the sites and the solid lines between them indicate the $XY$ interaction. The spins at red (blue) sites $r$ are entangled only with their partners $\overline{r}$ in the state $|\kappa_{S}^{-}\rangle$ ($|\kappa_{S}^{+}\rangle$) [defined in Eq.~\eqref{SMeq: kappa}].}
    \label{fig: supp_fig6}
\end{figure}
\subsection{Thermality of Entanglement}
\begin{figure}
    \centering
    \includegraphics[scale=0.4]{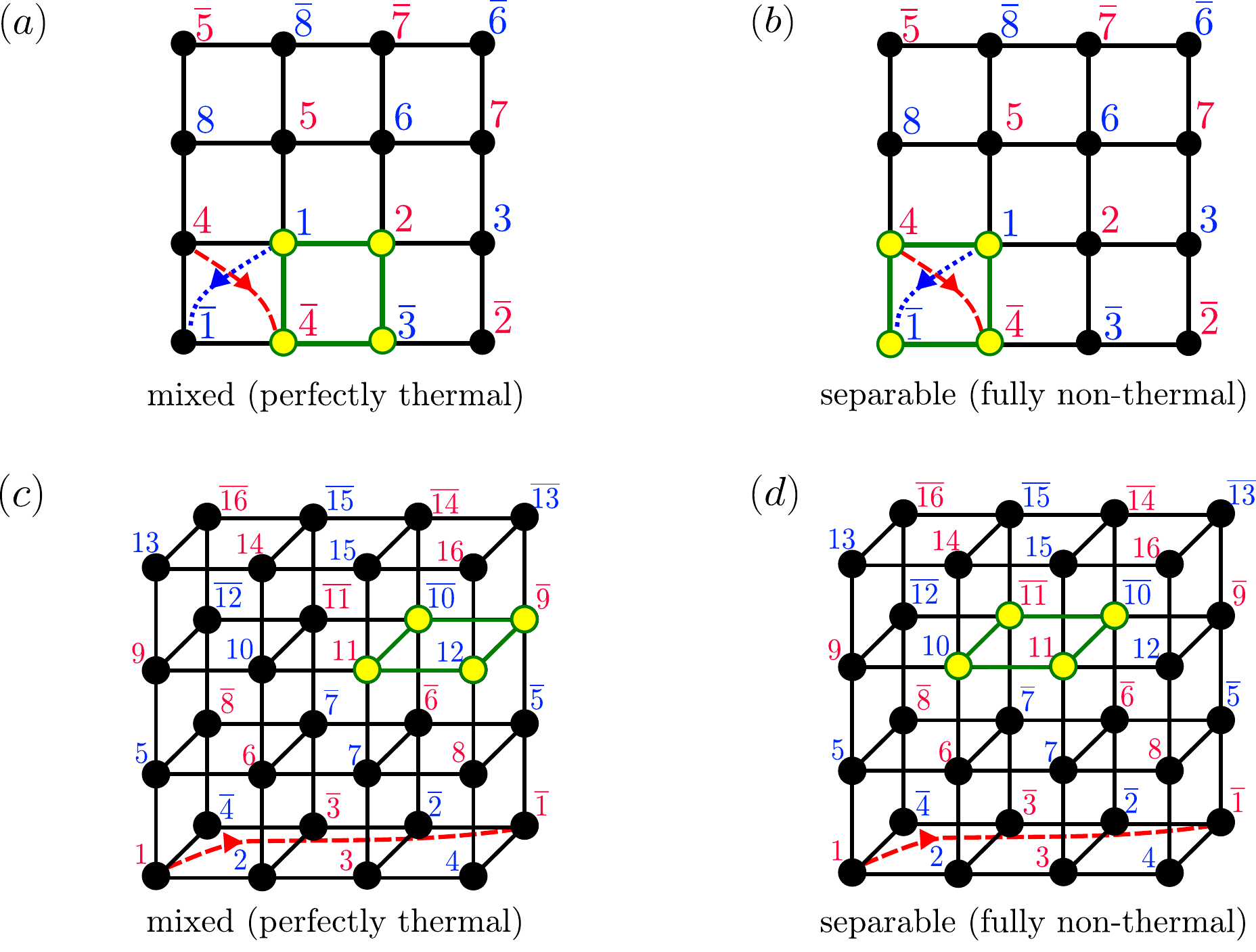}
    \caption{Schematic representation of volume-entangled thermal (left panels) and separable non-thermal (right panels) bipartitions of zero energy eigenstates of the nearest neighbor spin-$S$ $XY$ model with a transverse field. (a) and (b) correspond to an eigenstate on a $M{\times}2L{=}4{\times}4$ square lattice with PBC, while (c) and (d) correspond to an eigenstate on a $M{\times}2L{\times}N{=}4{\times}4{\times}2$ cuboid with PBC [(c) and (d)]. The black dots denote the sites and the solid lines between them indicate the $XY$ interaction. The spins at red (blue) sites $r$ are entangled only with their partners $\overline{r}$ in the state $|\kappa_{S}^{-}\rangle$ ($|\kappa_{S}^{+}\rangle$) [defined in Eq.~\eqref{SMeq: kappa}]. In each case, the bipartition is such that the yellow-colored sites connected by green bonds form a subsystem, and the remaining sites form the complementary subsystem.}
    \label{fig: supp_fig7}
\end{figure}
Note that, unlike in one dimension, not all contiguous bipartitions in higher dimensions exhibit the ``thermal" volume-law scaling of entanglement. For example, consider one of the zero energy eigenstates of $M{\times}2L{=}4{\times}4$ square lattice with PBC-the Type-3 (vii) state as shown in Fig.~\ref{fig: supp_fig5}.

In Fig.~\ref{fig: supp_fig7}(a), we show a contiguous bipartition where the subsystem (yellow-colored sites connected by green bonds) exhibits volume-law entanglement and remains maximally mixed or perfectly thermal. In contrast, Fig.~\ref{fig: supp_fig7}(b) shows a bipartition where the same state becomes separable and non-thermal (subsystem consisting entirely of the top or bottom right squares, or the top or bottom left squares are separable). Similarly, in Fig.~\ref{fig: supp_fig7}(c) and (d), we provide examples of volume-entangled thermal and separable non-thermal bipartitions for an eigenstate on a $M{\times}2L{\times}N{=}4{\times}4{\times}2$ cuboid with PBC in all directions. Intermediate behavior, i.e., neither maximally entangled nor fully separable is also possible, for example, for the state shown in Fig.~\ref{fig: supp_fig7}(a), if the bipartition results in a subsystem consisting of the sites $\{1,\bar{1}, 4, \bar{4}, 5, 8\}$ and the rest of the sites form its complement, there will be finite but non-maximal entanglement between the partitioned regions. These cases further illustrate how the entanglement structure in higher dimensions depends sensitively on the choice of bipartition and the pairing pattern of the eigenstate. A detailed characterization of how many such contiguous bipartitions fail to be maximally mixed in the different kinds of zero-energy eigenstates in higher dimensions, and what their implications are for thermalization, is an interesting question for future explorations.

\section{Reproducing previously constructed eigenstates of various spin models}
\label{sec: previous_examples}
In this section, we apply our formalism to previously studied spin-$1/2$ Hamiltonians in the literature to demonstrate that it can reproduce the same zero energy states as reported in those works. Note that for all of the following examples, we consider the Hilbert space defined by computational basis $|\vec{b}\rangle{=}{\bigotimes}_{i{=}1}^{L}|b_{i}\rangle$ $(\mathcal{H}_{L}{\equiv}\text{span}\{|\vec{b}\rangle\})$, where $\vec{b}{=}(b_{1}, b_{2}, {\cdots}, b_{L})$ is either a string of spins, i.e., $b_j \in \{\uparrow, \downarrow\}$, or a ``bit-string" of length $L$, i.e., $b_j \in \{0, 1\}$.

\subsection{A model with three-spin interaction and transverse magnetic field~\cite{udupa2023weak}}

First, we consider the following Hamiltonian with a three-spin interaction and a transverse magnetic field for a system of size $L$ (with PBC) that was studied in Ref.~\cite{udupa2023weak}
\begin{equation}
\begin{split}
    H_{1}(L){=}\sum_{i{=}1}^{L}\big(h\sigma_{i}^{z}+J\sigma_{i}^{x}\sigma_{i+1}^{x}\sigma_{i+2}^{x}\big),
\end{split}
    \label{SMeq: three_spin_ising}
\end{equation}
where $h$ and $J$ are arbitrary constants. Here each $\hat{H}^{\alpha}_{i+j_{1},\cdots,i+j_{n}}$, i.e., $\sigma_{i}^{x}\sigma_{i+1}^{x}\sigma_{i+2}^{x}$ and $\sigma_{i}^{z}$ are real in the chosen basis and anti-commute with the operator $\mathcal{C}_{1}{=}\prod_{i}^{L} \sigma_{i}^{y}$ [hence they satisfy Eq.~\eqref{SMeq: condition}]. Thus using Theorem~\ref{SMtheorem: 1}, a zero energy eigenstate of Hamiltonian $H_{1}(2L)$ is given by
\begin{equation}
|\Lambda\rangle_{\mathcal{C}_{1}}{=}\bigotimes_{i{=}1}^{L}|\Phi^{-}\rangle_{i,\overline{i}},
\label{SMeq: Lambda_state1}
\end{equation}
where $|\Phi^{\pm}\rangle_{i,j}{=}(|01\rangle {\pm} |10\rangle)_{i,j}/\sqrt{2}$.
Similarly, when the half-chain length $L$ is a multiple of three, the operator $\mathcal{C}_{1}'{=}\prod_{p{=}1}^{\frac{L}{3}} \sigma_{3p{-}2}^{x}\sigma_{3p{-}1}^{y}\sigma_{3p}^{x}$ anti-commutes with each $\hat{H}_{i{+}j_{1}, {\cdots}, i{+}j_{n}}^{\alpha}$ [thus satisfying Eq.~\eqref{SMeq: condition}] resulting in another zero energy eigenstate of the form
\begin{equation}
    |\Lambda\rangle_{\mC_{1}'}{=}\bigotimes_{i{=}1}^{L{-}2}|\Phi^{+}\rangle_{i,\overline{i}}|\Phi^{-}\rangle_{i{+}1,\overline{i{+}1}}|\Phi^{+}\rangle_{i{+}2,\overline{i{+}2}}.
    \label{SMeq: Lambda_state1p}
\end{equation}
Note that due to PBC the translations of the state $|\Lambda\rangle_{\mC_{1}'}$ by one site, i.e., $|\Lambda^{\mathcal{T}_1}\rangle_{\mC_{1}'}{=}\bigotimes_{i{=}1}^{L{-}2}|\Phi^{-}\rangle_{i,\overline{i}}|\Phi^{+}\rangle_{i{+}1,\overline{i{+}1}}|\Phi^{+}\rangle_{i{+}2,\overline{i{+}2}}$ and two sites i.e. $|\Lambda^{\mathcal{T}_2}\rangle_{\mC_{1}'}{=}\bigotimes_{i{=}1}^{L{-}2}|\Phi^{+}\rangle_{i,\overline{i}}|\Phi^{+}\rangle_{i{+}1,\overline{i{+}1}}|\Phi^{-}\rangle_{i{+}2,\overline{i{+}2}}$ are also distinct zero energy eigenstates of $H_{1}(2L)$. The states of Eqs.~\eqref{SMeq: Lambda_state1} and \eqref{SMeq: Lambda_state1p} are exactly the volume-law entangled states reported in Ref.~\cite{udupa2023weak}. 

\subsection{PXP model~\cite{ivanov2024volume}}

Next, we reproduce the volume-law state constructed in the PXP model in Ref.~\cite{ivanov2024volume} using our framework. The PXP model for a spin chain of size $L$ with PBC is described by the Hamiltonian~\cite{Turner2018weak}
\begin{equation}
    H_{\rm PXP}(L)=\sum_{i{=}1}^{L} \hat{P}_{i{-}1}\sigma_{i}^{x}\hat{P}_{i{+}1}=\sum_{i{=}1}^{L} \Big[\sigma_{i}^{x}{+}\sigma_{i}^{x}\sigma_{i{+}1}^{z}{+}\sigma_{i{-}1}^{z}\sigma_{i}^{x}{+}\sigma_{i{-}1}^{z}\sigma_{i}^{x}\sigma_{i{+}1}^{z}\Big].
    \label{SMeq: PXP_Hamiltonian}
\end{equation}
Here, we used $\hat{P}_{i}{=}(\mathbb{I}{+}\sigma_{i}^{z})/2$ to write the Hamiltonian in the form given in Eq.~\eqref{SMeq: generic_Hamiltonian}. Now in our chosen computational basis the global particle-hole symmetry operator $\mathcal{C}_{\rm PXP}{=}\prod_{i{=}1}^{L}\sigma_{i}^{z}$ satisfies the conditions stated in Eq.~\eqref{SMeq: condition}. Thus, using Theorem~\ref{SMtheorem: 1}, a zero energy eigenstate of $H_{\rm PXP}(2L)$ is
\begin{equation}
    |\Lambda\rangle_{\rm PXP}{=}\bigotimes_{i{=}1}^{L}|\Psi^{-}\rangle_{i,\overline{i}}{=}
    \sum_{|\vec{b}\rangle\in\mathcal{H}_{L}} ({-}1)^{|b|}|\vec{b}\rangle_{1,\cdots,L}  \otimes |\vec{b}\rangle_{\overline{1},\cdots,\overline{L}},
    \label{SMeq: Lambda_PXP}
\end{equation}
where we have used $\mC_{\rm PXP}|\vec{b}\rangle{=}({-}1)^{|b|}|\vec{b}\rangle$ with $|b|$ denoting the parity of the bit string $\vec{b}$ and $|\Psi^{-}\rangle_{i,j}{=}(|00\rangle {-} |11\rangle)_{i,j}/\sqrt{2}$

While $\ket{\Lambda}_{\rm PXP}$ is an eigenstate of the PXP model in the full $2^L$-dimensional Hilbert space, the Hilbert space of the PXP model is usually truncated to the so-called nearest-neighbor Rydberg blockade Hilbert space~\cite{Turner2018weak}. This truncated blockade subspace $\mathcal{F}_{L}$ for chains of $L$ spins with PBC is spanned by basis vectors $|\vec{f}\rangle$ whose bit-strings $\vec{f}$ do not contain the sub-bit-string ``11". The projector on this subspace defined by $\mathcal{P}_{\text{Ryd}}{=}\prod_{i}(\mathbb{I}{-}\hat{Q}_{i}\hat{Q}_{i{+}1})$ (i.e., $\mathcal{P}_{\text{Ryd}}\mH_{L}{=}\mathcal{F}_{L})$, where $\hat{Q}_{i}{=}\mathbb{I}{-}\hat{P}_{i}$, commutes with the Hamiltonian $H_{\rm PXP}$, which enables the truncation of this Hilbert space for the PXP Hamiltonian. Thus, the projection of the state $|\Lambda\rangle_{\rm PXP}$ into the Rydberg blockade sector should also be an eigenstate of the PXP model truncated into this Hilbert space, and this state reads
\begin{equation}
    |\Lambda\rangle_{\rm Ryd}{=}\mathcal{P}_{\text{Ryd}}|\Lambda\rangle_{\rm PXP}{=}\sum_{|\vec{f}\rangle\in \mathcal{F}_{L}}({-}1)^{|f|}|\vec{f}\rangle_{1,\cdots, L}\otimes |\vec{f}\rangle_{L{+}1,\cdots, 2L}.
    \label{SMeq: PXP_eigenstate_Rydberg_subspace}
\end{equation}
This is precisely the volume-law eigenstate reported in Ref.~\cite{ivanov2024volume} for the PXP chain. We note that this construction also works for the PXP model in arbitrary dimensions and in a wide variety of geometries, which then reproduces the results of Ref.~\cite{ivanov2024volume}.

\subsection{Entangled antipodal paired (EAP) state~\cite{chiba2024exact}}

Next, we apply our formalism to a next-nearest-neighbor interacting Hamiltonian for a system of size $L$ with PBC considered in Ref.~\cite{chiba2024exact} given by
\begin{equation}
\begin{split} 
    H_{\rm NNN}(L){=}\sum_{i{=}1}^{L} \Big( J^{xzx}\sigma_{i}^{x}\sigma_{i{+}1}^{z}\sigma_{i{+}2}^{x}{+}J^{yzy}\sigma_{i}^{y}\sigma_{i{+}1}^{z}\sigma_{i{+}2}^{y}{+}J^{xzy}\sigma_{i}^{x}\sigma_{i{+}1}^{z}\sigma_{i{+}2}^{y}{+}J^{yzx}\sigma_{i}^{y}\sigma_{i{+}1}^{z}\sigma_{i{+}2}^{x}{+}\\
J^{zzz}\sigma_{i}^{z}\sigma_{i{+}1}^{z}\sigma_{i{+}2}^{z}{+}
J^{xx}\sigma_{i}^{x}\sigma_{i{+}1}^{x}{+}
J^{yy}\sigma_{i}^{y}\sigma_{i{+}1}^{y}{+}
J^{xy}\sigma_{i}^{x}\sigma_{i{+}1}^{y}{+}
J^{yx}\sigma_{i}^{y}\sigma_{i{+}1}^{x}{+}
J^{z}\sigma_{i}^{z}\Big),
\end{split}
\end{equation}
where all $J$'s are arbitrary. In our chosen computational basis, the Hamiltonian consists of terms $\hat{H}^{\alpha}_{i+j_{1},\cdots,i+j_{n}}$ that are either purely real (i.e., $\sigma_{i}^{x}\sigma_{i{+}1}^{z}\sigma_{i{+}2}^{x}$, $\sigma_{i}^{y}\sigma_{i{+}1}^{z}\sigma_{i{+}2}^{y}$, $\sigma_{i}^{z}\sigma_{i{+}1}^{z}\sigma_{i{+}2}^{z}$, $\sigma_{i}^{x}\sigma_{i{+}1}^{x}$, $\sigma_{i}^{y}\sigma_{i{+}1}^{y}$ and $\sigma_{i}^{z}$) or purely imaginary (i.e., $\sigma_{i}^{x}\sigma_{i{+}1}^{z}\sigma_{i{+}2}^{y}$, $\sigma_{i}^{y}\sigma_{i{+}1}^{z}\sigma_{i{+}2}^{x}$,  
$\sigma_{i}^{x}\sigma_{i{+}1}^{y}$ and $\sigma_{i}^{y}\sigma_{i{+}1}^{x}$). When $L$ is even, all real $\hat{H}^{\alpha}_{i+j_{1},\cdots,i+j_{n}}$ {\it anticommute} with an operator $\mathcal{C}_{\rm NNN}{=}\prod_{p{=}1}^{\frac{L}{2}}\sigma_{2p{-}1}^{x}\sigma_{2p}^{y}$ while all purely imaginary $\hat{H}^{\alpha}_{i+j_{1},\cdots,i+j_{n}}$ {\it commute} with $\mathcal{C}_{\rm NNN}$ satisfying the condition given in Eq.~\eqref{SMeq: condition}. Thus, as a result of Theorem~\ref{SMtheorem: 1}, the Hamiltonian $H_{\rm NNN}(2L)$ hosts a zero energy eigenstate given by
\begin{equation}
    |\Lambda\rangle_{\rm NNN}{=}\bigotimes_{i{=}1}^{L{-}1}|\Phi^{-}\rangle_{i,\overline{i}}|\Phi^{+}\rangle_{i{+}1,\overline{i{+}1}},
    \label{SMeq: EAP_eigenstate}
\end{equation}
where $|\Phi^{\pm}\rangle_{i,j}{=}(|01\rangle {\pm} |10\rangle)_{i,j}/\sqrt{2}$. Note that due to PBC, the translation of the state by one site, given by 
\begin{equation}
|\Lambda^{\mathcal{T}_1}\rangle_{\rm NNN}{=}\bigotimes_{i{=}1}^{L{-}1}|\Phi^{+}\rangle_{i,\overline{i}}|\Phi^{-}\rangle_{i{+}1,\overline{i{+}1}}   
\end{equation}
produces another distinct zero energy eigenstate (orthogonal to $|\Lambda\rangle_{\rm NNN}$). These are precisely the entangled antipodal paired (EAP) states [named so since there are Bell-pairs between antipodal points when the lattice sites are arranged uniformly on a circle] identified in Ref.~\cite{chiba2024exact}.

\end{document}